\documentclass[preprint,12pt]{elsarticle}

\usepackage{amssymb}

\journal{arXiv}

\usepackage{amsmath}
\usepackage{amssymb,amsfonts}
\usepackage{bm}
\usepackage[ruled,vlined,linesnumbered,commentsnumbered]{algorithm2e}
\usepackage{enumitem}

\newtheorem{theorem}{Theorem}
\newtheorem{lemma}[theorem]{Lemma}
\newtheorem{corollary}[theorem]{Corollary}
\newtheorem{claim}[theorem]{Claim}
\newtheorem{proposition}[theorem]{Proposition}
\newtheorem{remark}[theorem]{Remark}
\newproof{proof}{Proof}

\newcommand{\textem}[1] {{\em #1}}
\newcommand\msize[1]{\left|#1\right|}

\newcommand\mset[1]{\{#1\}}
\newcommand\mseq[1]{\langle#1\rangle}
\newcommand\mvec[1]{\left(#1\right)}
\newcommand\order[1]{\mathcal{O}(#1)}
\newenvironment{mycase}[2]{%
\vspace{3mm} \noindent \textbf{Case~#1}: #2\par}{}

\newenvironment{mycaselast}[2]{%
\vspace{3mm} \noindent \textbf{Case~#1}: #2\par}{\vspace{3mm}}

\newenvironment{listing}[1]{%
        \begin{list}{*}{%
                 \settowidth{\labelwidth}{#1}%
                 \setlength{\leftmargin}{\labelwidth}%
                  \advance \leftmargin by 12pt
                   \setlength{\itemsep}{0pt}%
                   \setlength{\parsep}{0pt}%
                   \setlength{\topsep}{0pt}%
                   \setlength{\parskip}{0pt}%
}%
}{%
\end{list}}

\newcommand{\llset}{\calL}
\newcommand{\llsetopt}{\calL^\mathrm{opt}}
\newcommand{\llsetone}{\calL^1}

\newcommand{\inv}{\operatorname{inv}}
\newcommand{\DV}{\mathit{DV}}
\newcommand{\pline}{\mathrm{pl}}
\newcommand{\plset}{\mathrm{PL}}
\newcommand{\cross}{\mathrm{cr}}
\newcommand{\crset}{\mathrm{CR}} 

\newcommand{\parent}{\mathrm{par}}

\newcommand{\lefttangled}{\mathit{LT}}

\newcommand{\Z}{\mathbb{Z}}
\newcommand{\frakS}{\mathfrak{S}}
\newcommand{\calL}{\mathcal{L}}

\newcommand{\optdv}{\mathrm{opt}_{\mathit{DV}}}
\newcommand{\optclldv}{\mathrm{opt}_{\mathit{CLL}}}
\newcommand{\symdiff}[2]{#1\bigtriangleup#2}
\newcommand{\symdiffDV}[2]{#1\bigtriangleup_{\mathit{DV}}#2}

\newcommand{\setodv}{\mathcal{D}}
\newcommand{\rootdv}{\bm{x_0}}

\newcommand{\parentdv}{\mathrm{par}}
\newcommand{\minidx}{\alpha}
\newcommand{\maxidx}{\beta}

\newcommand{\mmcont}{T} 

\newenvironment{myproblem}[3]{%
\vspace{3mm} 
\noindent \textbf{Problem:}#1\par%
\noindent \textbf{Instance:}#2\par%
\noindent \textbf{Question:}#3}{\vspace{3mm}}

\newcommand{\reconfCLLDV}{\textsc{ReconfCLL-DV}}

\newcommand{\reconfDV}{\textsc{ReconfDV}}

\newcommand{\enumCLLDV}{\textsc{EnumCLL-DV}}
\newcommand{\enumCLL}{\textsc{EnumCLL}}
\newcommand{\enumDV}{\textsc{EnumDV}}

\usepackage{color}
\definecolor{darkred}{rgb}{0.8,0,0}

\begin{document}

\begin{frontmatter}

\title{Reconfiguration and Enumeration of\\ 
Optimal Cyclic Ladder Lotteries\tnoteref{tnote}}
\tnotetext[tnote]{A preliminary version of this article appeared in IWOCA2023~\cite{NozakiWY23}.}

\author[yokohama,hiroshima]{Yuta Nozaki}
\ead{nozaki-yuta-vn@ynu.ac.jp}
\author[hosei]{Kunihiro Wasa}
\ead{wasa@hosei.ac.jp}
\author[iwate]{Katsuhisa Yamanaka}
\ead{yamanaka@iwate-u.ac.jp}

\affiliation[yokohama]{organization={Yokohama National University},
            country={Japan}}
\affiliation[hiroshima]{organization={WPI-SKCM$^2$, Hiroshima University}, 
country={Japan}}
\affiliation[hosei]{organization={Hosei University}, country={Japan}}
\affiliation[iwate]{organization={Iwate University}, country={Japan}}

\begin{abstract}
A \emph{ladder lottery}, known as ``Amidakuji'' in Japan, is a common way to decide an assignment at random.
In this paper, we investigate reconfiguration and enumeration problems of cyclic ladder lotteries.
First, when a permutation $\pi$ and an optimal displacement vector $\bm{x}$ are given, we investigate the reconfiguration and enumeration problems of the ``optimal'' cyclic ladder lotteries of $\pi$ and $\bm{x}$.
Next, for a give permutation $\pi$ we consider reconfiguration and enumeration problems of the optimal displacement vectors of $\pi$.
\end{abstract}


\end{frontmatter}

\section{Introduction}
\label{sec:intro}

A \emph{ladder lottery}, known 
as ``Amidakuji'' in Japan,
is a common way to decide an assignment at random.
Formally, we define ladder lotteries as follows.
A \emph{network} is a sequence $\mseq{\ell_1,\ell_2, \ldots ,\ell_n}$ of $n$ vertical lines~(\emph{lines} for short) and horizontal lines~(\emph{bars} for short) each of which connects two consecutive vertical lines.
We say that $\ell_i$ is located on the left of $\ell_j$ if $i<j$ holds.
The $i$-th line from the left is called the \emph{line $i$}.
We denote by $[n]$ the set $\mset{1,2,\ldots, n}$.
Let $\pi = (\pi_1,\pi_2,\ldots ,\pi_n)$ be a permutation of $[n]$.
A \emph{ladder lottery} of $\pi$ is a network with $n$ lines and zero or more bars such that
\begin{listing}{aaa}
\item[(1)] the top endpoints of lines correspond to the identity permutation,
\item[(2)] each bar exchanges two elements in $[n]$, and
\item[(3)] the bottom endpoints of lines correspond to $\pi$.
\end{listing}
\noindent
See \figurename~\ref{fig:cyc_lad_lot}(a) for an example.
\begin{figure}[t]
  \vspace{-3mm}
  \centerline{\includegraphics[width=1.0\linewidth]{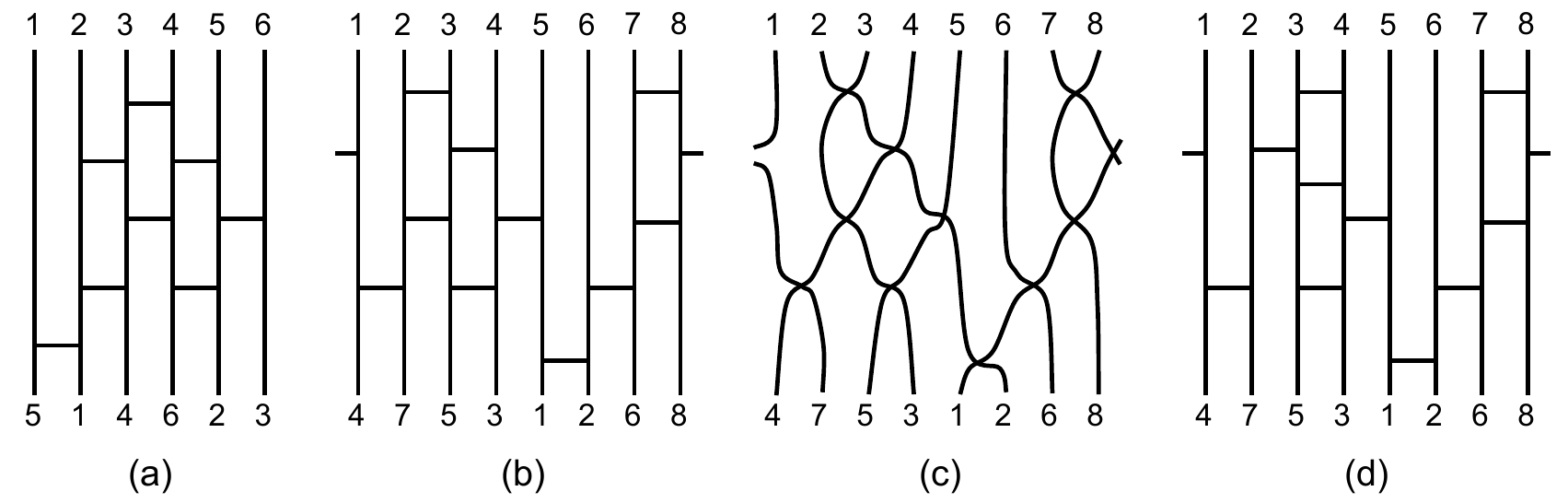}}
  \vspace{-3mm}
  \caption{(a) An optimal ladder lottery of the permutation $(5,1,4,6,2,3)$. (b) A cyclic ladder lottery of the permutation $(4,7,5,3,1,2,6,8)$ and (c) its representation as a pseudoline arrangement. (d) A cyclic ladder lottery obtained from (b) by applying a braid relation to the triple $(2,3,4)$.}
\label{fig:cyc_lad_lot}
\end{figure}
In each bar of a ladder lottery, two elements are swapped.
We can regard a bar as an adjacent transposition of two elements in the current permutation,
and the permutation always results in the given permutation $\pi$.
A ladder lottery of a permutation 
$\pi$ is \textem{optimal} if it consists of the minimum number of bars
among ladder lotteries of $\pi$.
Let $L$ be an optimal ladder lottery of $\pi$ 
and let $m$ be the number of bars in $L$. 
Then, we can observe that $m$ is equal to the number of ``inversions'' 
of $\pi$, which are pairs
$(\pi_i,\pi_j)$ in $\pi$ with $\pi_i > \pi_j$ and $i<j$.
The ladder lottery in \figurename~\ref{fig:cyc_lad_lot}(a) has 8 bars
and the permutation (5,1,4,6,2,3) has 8 inversions:
(5,1), (5,4), (5,2), (5,3), (4,2), (4,3), (6,2), and (6,3),
so the ladder lottery is optimal. 

The ladder lotteries are related to some objects
in theoretical computer science.
First, the ladder lotteries are strongly related
to primitive sorting networks,
which are investigated by Knuth~\cite{Knuth92}.
A primitive sorting network uses a ``comparator'' to exchange two elements instead of a bar in a ladder lottery.
A comparator exchanges two elements $\pi_i$ and $\pi_{j}$ if $\pi_i > \pi_j$ holds,
while a bar in a ladder lottery always exchanges them.
Next, the optimal ladder lotteries of the reverse identity permutation of $[n]$
one-to-one correspond to arrangements\footnote{An arrangement is \emph{simple} if no three pseudolines have a common intersection point. In this paper, we assume that all the pseudoline arrangement are simple.}
of $n$ pseudolines\footnote{A \emph{pseudoline} in the Euclidean plane is a $y$-monotone curve extending from positive infinity to negative infinity.
} 
such that any two pseudolines intersect exactly once~(see \cite{YamanakaNMUN10}).
Each bar in a ladder lottery corresponds to an intersection of two pseudolines.
Note that, in an optimal ladder lottery of the reverse identity permutation, each pair of two elements in $[n]$ is swapped exactly once on a bar. 
For such pseudoline arrangements,
there exist several results on the bounds of the number of them.
Let $B_n$ be the number of the arrangements of $n$ pseudolines such that any two pseudolines intersect exactly once and let $b_n = \log_2{B_n}$.
The problem of estimating $B_n$ is first considered by Knuth~\cite{Knuth92}.
He showed 
the following upper and lower bounds: $b_n \leq 0.7924(n^2+n)$ and $b_n \geq n^2/6 - \order{n}$.
The upper bound was improved to $b_n \leq 0.6974n^2$ by Felsner~\cite{F97} and $b_n \leq 0.6571n^2$ by Felsner and Valtr~\cite{FelsnerV11}.
The lower bound was improved to $b_n \geq 0.1887n^2$ by Felsner and Valtr~\cite{FelsnerV11}
and $b_n \geq 0.2053n^2$ by 
Dumitrecu and Mandal~\cite{DumitrescuM20}.
For optimal ladder lotteries of reverse identity permutations,
several exact-counting results are known for small $n$~\cite{KawaharaSYM11,Samuel11,YamanakaNMUN10}.
See On-Line Encyclopedia of Integer Sequences for further details~\cite{Sloane22}.

From a mathematical viewpoint, ladder lotteries with $n$ lines correspond to (positive) words\footnote{Here, a word is a sequencen of transpositions.} of the $n$-th symmetric group $\frakS_n$ with the following presentation
\[
\langle s_1,s_2,\dots,s_{n-1} \mid s_i^2=1,\ s_is_j=s_js_i\ (|i-j|\geq 2),\  s_is_{i+1}s_i=s_{i+1}s_is_{i+1} \rangle.
\]
The third relation is called a braid relation.
In particular, the word length is equal to the number of bars in a ladder lottery.
See \cite{MKS04}.

In this paper, we 
consider a variation of ladder lotteries, ``cyclic'' ladder lotteries.
A \emph{cyclic ladder lottery} is a ladder lottery that is allowed to have bars between the first and last lines.
\figurename~\ref{fig:cyc_lad_lot}(b) shows an example of a cyclic ladder lottery.
Similar to ladder lotteries,
cyclic ladder lotteries with $n$ lines correspond to the mathematical objects, which are
(positive) words of the $n$-th affine symmetric group $\widetilde\frakS_n$ with the following presentation
\[
\langle s_1,s_2,\dots,s_{n-1},s_n \mid s_i^2=1,\ s_is_j=s_js_i\ (|i-j|\geq 2),\ s_is_{i+1}s_i=s_{i+1}s_is_{i+1} \rangle,
\]
where the indices in the second and third relations are considered as modulo $n$.
Lusztig~\cite{Lus83} mentioned this group from a viewpoint of the theory of Coxeter groups.

Now, as is the case with ladder lotteries, we introduce ``optimality'' to cyclic ladder lotteries. A cyclic ladder lottery of a permutation $\pi$ is \emph{optimal} if it has the minimum number of bars.
It is known that the 
minimum number of bars in a cyclic ladder lottery of a permutation is equal to the 
cyclic analogue of inversion number
and can be computed in $\order{n^2}$ time~\cite{Jerrum85}.

For the optimal ladder lotteries of a permutation $\pi$, reconfiguration and enumeration problems have been solved~\cite{YamanakaHW21,YamanakaNMUN10}.
The key observation 
is that reconfiguration graphs under braid relations are always connected, where a reconfiguration graph is a graph such that each vertex corresponds to an optimal ladder lottery of $\pi$ and each edge corresponds to a braid relation between two optimal ladder lotteries.
Hence, for the reconfiguration problems, the answer to a reachability problem is always yes.
Moreover, Yamanaka~et~al.~\cite{YamanakaHW21} characterized the length of a shortest reconfiguration sequence and proposed an algorithm that finds it.
For the enumeration problem, Yamanaka~et~al.~\cite{YamanakaNMUN10} designed an algorithm that enumerates them by traversing a spanning tree defined on a reconfiguration graph.
Now, does the same observation of reconfiguration graphs hold for optimal cyclic ladder lotteries?
Can we solve reconfiguration and enumeration problems for optimal cyclic ladder lotteries?

For optimal cyclic ladder lotteries,
reconfiguration graphs under braid relations may be disconnected.
For example, the two optimal cyclic ladder lotteries of the permutation $(4,2,6,1,5,3)$ in \figurename~\ref{fig:amida_diff_vecs} have no reconfiguration sequence under braid relations.
Actually, it can be observed that the set of optimal cyclic ladder lotteries of a permutation $\pi$ is partitioned into the sets of optimal cyclic ladder lotteries with the same
``optimal displacement vectors''~\cite{Jerrum85}, which represent the movement direction of the each element in $[n]$ in optimal cyclic ladder lotteries.
Note that applying a braid relation does not change a displacement vector.
Therefore, to enumerate all the optimal cyclic ladder lotteries of a permutation $\pi$,
we have to solve two enumeration problems: (1) enumerate all the optimal cyclic ladder lotteries of $\pi$ with the same optimal displacement vector and (2) enumerate all the optimal displacement vectors of $\pi$.
We first consider reconfiguration and enumeration problems for cyclic optimal ladder lotteries of a given permutation $\pi$ and optimal displacement vector $\bm{x}$.
For the reconfiguration problem, we show that any two optimal cyclic ladder lotteries of $\pi$ and $\bm{x}$ are always reachable under braid relations and investigate a shortest reconfiguration sequence of them.
Then, for the enumeration problem, we design an algorithm that enumerates all the cyclic optimal ladder lotteries of $\pi$ and  $\bm{x}$ in 
polynomial delay.
Next, we consider 
the two problems for the optimal displacement vectors of a permutation $\pi$ under ``max-min contractions'', which is an operation to optimal displacement vectors.
For the reconfiguration problem, we characterize the length of a shortest reconfiguration sequence between two optimal displacement vectors of $\pi$ and show that one can compute a shortest reconfiguration sequence.
For the enumeration problem, we design
a constant-delay algorithm that enumerates all the optimal displacement vectors of $\pi$.
\begin{figure}[tb]
  \centerline{\includegraphics[width=0.50\linewidth]{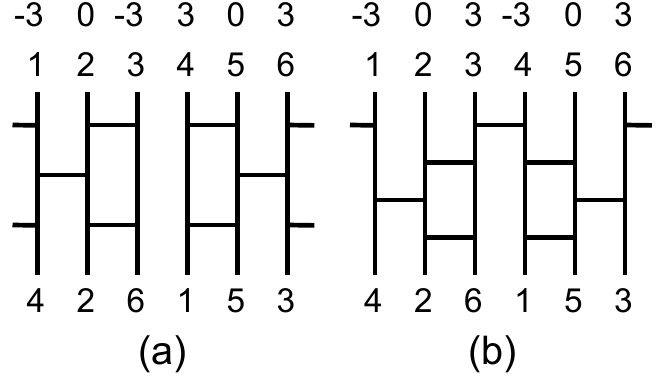}}
  \vspace{-3mm}
  \caption{(a) An optimal cyclic ladder lottery of the permutation $(4,2,6,1,5,3)$ and its optimal displacement vector $(-3,0,-3,3,0,3)$ and (b) an optimal cyclic ladder lottery of the same permutation and its optimal displacement vector $(-3,0,3,-3,0,3)$.}
  \vspace{-3mm}
\label{fig:amida_diff_vecs}
\end{figure}

\section{Preliminary}
\label{sec:pre}

A \emph{network} is a sequence $\mseq{\ell_1,\ell_2, \ldots ,\ell_n}$ of $n$ vertical lines~(\emph{lines} for short) and horizontal lines~(\emph{bars} for short) each of which connects 
the two lines in a pair in $\mset{(\ell_1,\ell_2), (\ell_2,\ell_3), \ldots , (\ell_{n-1},\ell_n),(\ell_n,\ell_1)}$.
We say that $\ell_i$ is located on the left of $\ell_j$ if $i<j$ holds.
The $i$-th line from the left is called the \emph{line $i$}.
We denote by $[n]$ the set $\mset{1,2,\ldots, n}$.
Let $\pi = (\pi_1,\pi_2,\ldots ,\pi_n)$ be a permutation of $[n]$.
A \emph{cyclic ladder lottery of $\pi$} is a network with $n$ lines and zero or more bars such that
\begin{listing}{aaa}
\item[(1)] the top endpoints of lines correspond to the identity permutation,
\item[(2)] each bar exchanges (cyclically adjacent) two elements in the current permutation, and
\item[(3)] the bottom endpoints of lines correspond to $\pi$.
\end{listing}
In a cyclic ladder lottery of $\pi$, each element $i$ in $[n]$ starts at the top endpoint of the line $i$, and goes down along the line, then whenever $i$ comes to an endpoint of a bar, $i$ goes to the other endpoint and goes down again, then finally $i$ reaches the bottom endpoint of the line $j$, where $i = \pi_j$.
This path is called the \textem{route} of the element $i$.
Each bar is regarded as a cyclically adjacent transposition and the composition of all the transpositions in a cyclic ladder lottery
always results in $\pi$.
A cyclic ladder lottery of $\pi$ is \emph{optimal} if the ladder lottery
contains the minimum number of bars.
For example, 
the cyclic ladder lottery in \figurename~\ref{fig:cyc_lad_lot}(b) is optimal.
We can observe from the figure that there exists an optimal cyclic ladder lottery $L$ without a ``slit'', that is, $L$ does not come from any ladder lottery.

A cyclic ladder lottery of $\pi$ is regarded as an arrangement of $n$ pseudolines on a cylinder.
The route of an element in $[n]$ corresponds to a pseudoline and a bar corresponds to an intersection of two pseudolines.
\figurename~\ref{fig:cyc_lad_lot}(c) shows the arrangement of pseudolines corresponding to the cyclic ladder lottery in \figurename~\ref{fig:cyc_lad_lot}(b).
We use terminologies on pseudoline arrangements on a cylinder instead of the ones 
on cyclic ladder lotteries to clarify discussions.
Let $\pline(L,i)$ denote the pseudoline of $i \in [n]$ in a cyclic ladder lottery $L$.
Note that the top endpoint of $\pline(L,i)$ corresponds to 
the element $i$ in the identity permutation, the bottom endpoint of $\pline(L,i)$ corresponds to 
$\pi_i$ in $\pi$, and $\pline(L,i)$ is an $y$-monotone curve.
In an optimal cyclic ladder lottery $L$, any two pseudolines cross at most once.
From now on, we assume that any two pseudolines in $L$ cross at most once.
For two distinct elements $i,j \in [n]$, 
$\cross(i,j)$ denotes the intersection of $\pline(L,i)$ and $\pline(L,j)$ if it exists.
For distinct $i,j,k \in [n]$, a triple $\mset{i,j,k}$
is \emph{tangled}
if $\pline(L,i)$, $\pline(L,j)$, and $\pline(L,k)$ cross each other.
Let $\mset{i,j,k}$ be a tangled triple in $L$.
Let $M$ be the ladder lottery induced from $L$ by the three pseudolines $\pline(L,i)$, $\pline(L,j)$, and $\pline(L,k)$,
and let $p,q,r$ be the three intersections in $M$.
Without loss of generality, in $M$,
we suppose that (1) $p$ is adjacent to the two top endpoints of two pseudolines,
(2) $q$ is adjacent to the top endpoint of a pseudoline and the bottom endpoint of a pseudoline, (3) $r$ is adjacent to the two bottom endpoints of two pseudolines.
Then, $\mset{i,j,k}$ is a \emph{left tangled triple} 
if $p,q,r$ appear in counterclockwise order on the contour of the region enclosed by $\pline(M,i)$, $\pline(M,j)$, and $\pline(M,k)$.
Similarly, $\mset{i,j,k}$ is a \emph{right tangled triple} 
if $p,q,r$ appear in clockwise order on the contour of the region.
See \figurename~\ref{fig:left_right_tangled} for examples.
\begin{figure}[tb]
  \centerline{\includegraphics[width=0.5\linewidth]{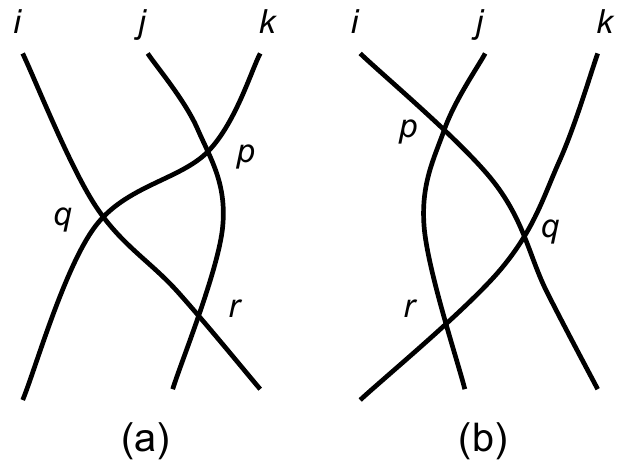}}
  \caption{Illustrations of (a) a left tangled triple and (b) a right tangled triple.}
\label{fig:left_right_tangled}
\end{figure}
A tangled triple $\mset{i,j,k}$ is \emph{minimal} if the region enclosed by $\pline(L,i)$, $\pline(L,j)$, and $\pline(L,k)$ includes no subpseudoline in $L$.
For example, in \figurename~\ref{fig:cyc_lad_lot}(c),
$\mset{2,3,4}$ and $\mset{7,8,1}$ are minimal right tangled triples
and 
$\mset{2,3,5}$ is a (non-minimal) right tangled triple.
A \emph{braid relation} is an operation to transform a minimal left (resp.\ right) tangled triple into another minimal right (resp.\ left) tangled triple.
The cyclic ladder lottery in \figurename~\ref{fig:cyc_lad_lot}(d) is obtained from the ladder lottery in \figurename~\ref{fig:cyc_lad_lot}(b) by applying a braid relation to the triple 
$\mset{2,3,4}$.

Let $\llset(\pi)$ be the set of the cyclic ladder lotteries of $\pi$.
Let $\llsetopt(\pi)$ be the set of the optimal cyclic ladder lotteries of $\pi$
and let $\llsetone(\pi)$ be the set of the cyclic ladder lotteries of $\pi$ such that any two pseudolines cross at most once.
Note that $\llsetopt(\pi) \subseteq \llsetone(\pi) \subseteq \llset(\pi)$ holds.
Let $L'$ be the ladder lottery obtained from $L \in \llsetopt(\pi)$ by removing $\pline(L,i)$ for $i \in [n]$.
Then, $L'$ is a cyclic ladder lottery in $\llset(\pi')$,
where $\pi'$ is the permutation obtained from $\pi$ by removing $i$. Since any two pseudolines in $L'$ cross at most once, $L'\in \llsetone(\pi')$ holds.
Note that $L'$ may be a non-optimal cyclic ladder lottery of $\pi'$.

The vector $\bm{x} = \mvec{x_1, x_2, \ldots ,x_n}$ is a \emph{displacement vector} of $\pi$ if $\sum_{i\in[n]}x_i=0$ and 
$i=\pi_{(i+x_i)\bmod n}$ for any $i\in[n]$.
Let $L$ be a cyclic ladder lottery in $\llset(\pi)$.
Then, a displacement vector can be defined from $L$ and is denoted by $\DV(L)$.
Intuitively, the element $x_i$ in $\DV(L) = \mvec{x_1, x_2, \ldots ,x_n}$ represents the movement direction of the element $i$ in $[n]$.
That is, if $x_i > 0$, the element $i$ goes right and if $x_i < 0$, the element $i$ goes left.
For instance, the displacement vector of the ladder lottery in \figurename~\ref{fig:cyc_lad_lot}(b) is $\mvec{-4,4,1,-3,-2,1,3,0}$.
A displacement vector $\bm{x}$ is \emph{optimal} if there exists an optimal cyclic ladder lottery $L \in \llsetopt(\pi)$
such that $\bm{x}=\DV(L)$ holds.
Similarly, a displacement vector $\bm{x}$ is said to be \emph{almost optimal} if there exists a cyclic ladder lottery $L$ of $\pi$ such that $L \in \llsetone(\pi,\bm{x})$ and $\bm{x}=\DV(L)$ hold.
We define the sets of cyclic ladder lotteries with the same displacement vector, as follows:
$\llsetopt(\pi,\bm{x}) = \mset{L \in \llsetopt(\pi) \mid \DV(L) = \bm{x}}$
and
$\llsetone(\pi,\bm{x}) = \mset{L \in \llsetone(\pi) \mid \DV(L) = \bm{x}}$.
Then, we have the following lemma which shows that the set $\llsetopt(\pi)$ is partitioned into sets of optimal cyclic ladder lotteries with the same optimal displacement vectors.

\begin{lemma}\label{lem:equiv_class}
Let $\pi$ be a permutation in $\frakS_n$.
Then,
\[
\llsetopt(\pi) = \bigsqcup_{\bm{x} \in \setodv(\pi)} \llsetopt(\pi,\bm{x}),
\]
where $\setodv(\pi)$ is the set of the optimal displacement vectors of $\pi$.
\end{lemma}

In the study of ladder lotteries, the inversion number plays a crucial role.
In \cite[(3.6)]{Jerrum85}, the cyclic analogue of the inversion number
\[
\inv(\bm{x})=\frac{1}{2}\sum_{i,j \in [n]}|c_{ij}(\bm{x})|
\]
is introduced, where $\bm{x}$ is a displacement vector.
Here, a crossing number $c_{ij}(\bm{x})$ is defined as follows.
Let $i,j$ be two elements in $[n]$
and let $r=i-j$ and $s=(i+x_i)-(j+x_j)$.
Then, we define $c_{ij}(\bm{x})$ by 
\[
c_{ij}(\bm{x}) =
\begin{cases}
    \msize{\mset{k\in [r,s] \mid k \equiv 0\mod n}} & \text{if $r \leq s$,} \\
    -\msize{\mset{k\in [s,r] \mid k \equiv 0 \mod n}} & \text{if $s < r$.}
\end{cases}
\]
The number $\inv(\bm{x})$ coincides with the affine inversion number for $\widetilde\frakS_n$ (see \cite[Section~8.3]{BjBr05} for instance).
As mentioned in \cite{Jerrum85}, $\inv(\bm{x})$ is equal to the number of intersections between the $n$ pseudolines on the cylinder.
Note here that \cite[Lemma~3.6]{Jerrum85} corresponds to \cite[Proposition~8.3.1]{BjBr05}.

\subsection{Longest elements}
The number of intersections of pseudolines in a cyclic ladder lottery $L$ is denoted by $\inv(\DV(L))$.
In this subsection, we simply denote it by $\inv(L)$ if it is obvious from the context.
We investigate the permutation $\pi$ such that the number of intersections in an optimal cyclic ladder lottery of $\pi$ is maximized and the number of intersections of an optimal cyclic ladder lottery of $\pi$.
That is, for each $n$, we focus on the value 
$\max_{\pi \in \frakS_n}\min_{L \in \calL(\pi)}\inv(L)$ 
and $\pi$'s which attain this value.
The corresponding problem in ordinary ladder lotteries is well-known.
In fact, the maximum is $n(n-1)/2$, which is uniquely attained by the reverse permutation (the so-called longest element in theory of Coxeter groups).

We give an alternative proof of \cite[Theorem~5]{ZuylenBSY16}.\footnote{In \cite{ZuylenBSY16}, the corresponding claim is stated for a set of $n$ positive integers instead of a permutation of $[n]$.
Hence, the proof in \cite{ZuylenBSY16} extends some notions on a permutation of $[n]$ to a set of positive integers.
On the other hand, our description here is straightforward and attains a short proof.
}

\begin{proposition}
\label{prop:longest}
Let $\pi$ be a permutation in $\frakS_n$.
\begin{enumerate}[label=\textup{(\arabic*)}]
    \item If $n=2m-1$, then 
    $\min\{\inv(L)\mid L\in\calL(\pi)\}\leq m(m-1)$ 
    and the equality holds only for $\pi=(m+1,\dots,2m-1,1,\dots,m),(m,\dots,2m-1,1,\dots,m-1)$.
    \item If $n=2m$, then 
    $\min\{\inv(L)\mid L\in\calL(\pi)\}\leq m^2$ 
    and the equality holds only for $\pi=(m+1,\dots,2m,1,\dots,m)$.
\end{enumerate}
\end{proposition}

\begin{lemma}
\label{lem:longest}
For the permutations $(m+1,\dots,2m-1,1,\dots,m)$, $(m,\dots,2m-1,1,\dots,m-1)$ and $(m+1,\dots,2m,1,\dots,m)$, the corresponding optimal displacement vectors are respectively
\begin{align*}
&(\overbrace{m-1,\dots,m-1}^{m},\overbrace{-m,\dots,-m}^{m-1}),\
(\overbrace{m,\dots,m}^{m-1},\overbrace{-m+1,\dots,-m+1}^{m}),\\ 
&(\overbrace{m,\dots,m}^{m},\overbrace{-m,\dots,-m}^{m})
\end{align*}
up to permutation \textup{(}see Figure~\ref{fig:Longest}\textup{)}.
\end{lemma}

\begin{proof}
One can see that the three displacement vectors are optimal and give the three permutations, respectively.
Then, the other optimal displacement vectors are obtained by max-min contractions discussed in Section~\ref{subsec:Reconfiguration} below.
\end{proof}

\begin{figure}[ht]
 \centering
 \includegraphics[width=0.7\textwidth]{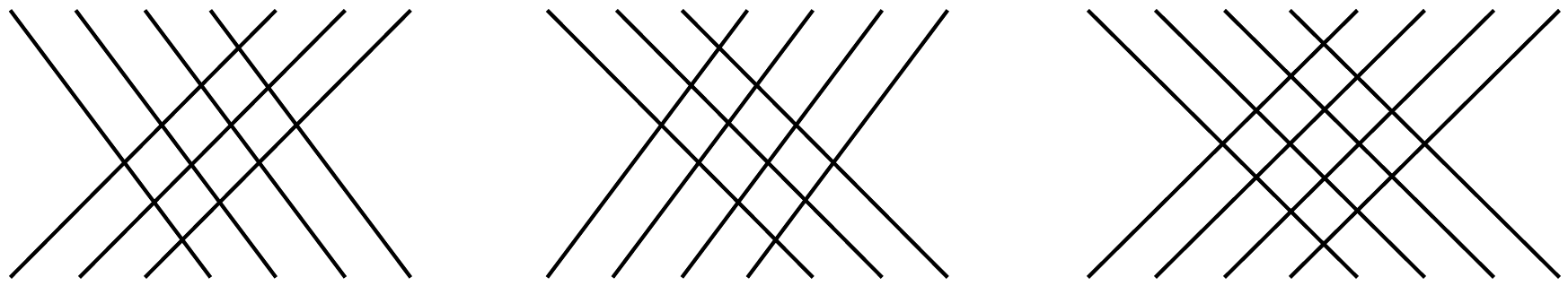}
 \caption{The pseudoline arrangements corresponding to the displacement vectors in Lemma~\ref{lem:longest} when $m=4$.}
 \label{fig:Longest}
\end{figure}

Define $c_{\min}(\bm{x})\in\Z$ by $c_{\min}(\bm{x})=\min_{j\in[n]}\sum_{i\in[n]}|c_{ij}(\bm{x})|$.
Then, we have $n c_{\min}(\bm{x})\leq 2\inv(\bm{x})$.
For the three displacement vectors in Lemma~\ref{lem:longest}, one sees that $c_{\min}(\bm{x})$ equals $m-1$, $m-1$ and $m$, respectively.
We simply write $c_{\min}(L)$ for $c_{\min}(\DV(L))$.

\begin{proof}[Proposition~\ref{prop:longest}]
We prove (1) and (2) simultaneously by induction on $m$.
The cases $m=1,2$ are obvious.
Suppose that both (1) and (2) hold for less than $m$.
Let 
$L \in \calL(\pi)$.
We may assume $L$ is optimal and $c_{\min}(L)=\sum_{i\in[n]}|c_{in}(L)|$.
We first show the inequality in (1).
Let $L'$ be the cyclic ladder lotteries obtained from $L$ by removing 
$\pline(L,n)$.
By subtracting $2c_{\min}(L)$ from the both sides of $(2m-1)c_{\min}(L)\leq 2\inv(L)$, one has $(2m-3)c_{\min}(L)\leq 2\inv(L')$.
Here, the induction hypothesis implies that $(2m-3)c_{\min}(L)\leq 2(m-1)^2$.
Since $m\geq 3$, we have $c_{\min}(L)\leq m-1$.
Therefore, one concludes that
\[
\inv(L)=\inv(L')+c_{\min}(L) \leq (m-1)^2+m-1 = m(m-1).
\]
We next consider $\pi$ for which the inequality is an equality.
This happens only if $\inv(L')=(m-1)^2$ and $c_{\min}(L)=m-1$.
By the induction hypothesis and Lemma~\ref{lem:longest}, the displacement vector $\bm{x}'=\DV(L')$ is determined up to permutation.
Here, $I=\{i\in[n-1]\mid |c_{in}(L)|=1\}$ coincides with $\{i\in[n-1]\mid x'_i=m-1\}$ or $\{i\in[n-1]\mid x'_i=-m+1\}$.
Indeed, if $i,j\in I$ with $x'_i=m-1$ and $x'_j=-m+1$, then $|x_i-x_j|>n$.
This contradicts the optimality of $L$.
Therefore, the corresponding permutation $\pi$ is one of the two in (1).
Similarly, we show (2) by the induction hypothesis and complete the induction argument.
\end{proof}

\begin{remark}
The reverse identity permutation $\pi$ of $[n]$ does not attain the minimum in Proposition~\ref{prop:longest}.
In fact, one can check that
\[
\min\{\inv(L)\mid L\in\calL(\pi)\}=
\begin{cases}
(m-1)^2 & \text{if $n=2m-1$,} \\
m(m-1)+1 & \text{if $n=2m$ and $m$ is odd,} \\
m(m-1) & \text{if $n=2m$ and $m$ is even.}
\end{cases}
\]
For instance, $(-1,-3,5,3,1,-1,-3,-5,3,1)$ is an optimal displacement vector of $\pi$ when $n=10$.
\end{remark}

\section{Reconfiguration and enumeration of cyclic ladder lotteries with optimal displacement vectors}

Let $\pi=(\pi_1,\pi_2,\ldots ,\pi_n)$ be a permutation 
in $\frakS_n$,
and let $\bm{x}=\mvec{x_1,x_2,\ldots ,x_n}$ be an optimal displacement vector of $\pi$.
In this section, 
we consider the problems of reconfiguration and enumeration for the set of the optimal cyclic ladder lotteries in $\llsetopt(\pi,\bm{x})$.

\subsection{Reachability}\label{sec:reach}

In this subsection, we consider a reconfiguration problem on reachability between two optimal cyclic ladder lotteries in $\llsetopt(\pi,\bm{x})$ under braid relations.
The formal description of the problem is given below.

\begin{myproblem}{
Reconfiguration of optimal cyclic ladder lotteries with optimal displacement vector~(\reconfCLLDV)}{
A permutation $\pi$, an optimal displacement vector $\bm{x}$ of $\pi$, and two optimal cyclic ladder lotteries $L,L' \in \llsetopt(\pi,\bm{x})$}{
Does there exist a reconfiguration sequence 
between $L$ and $L'$ under braid relations?}
\end{myproblem}

To solve \reconfCLLDV,
we consider the reconfiguration problem for the cyclic ladder lotteries in $\llsetone(\pi,\bm{y})$, where $\bm{y}$ is an almost optimal displacement vector of $\pi$.
We show that any two cyclic ladder lotteries in $\llsetone(\pi,\bm{y})$ are always reachable.
As a byproduct, we obtain an answer to \reconfCLLDV.

Let $L$ be a cyclic ladder lottery in $\llsetone(\pi,\bm{y})$ and suppose that $L$ contains one or more intersections.
Then, in $L$, there exists an element $i \in [n]$ such that $\pline(L,i)$ and $\pline(L,(i+1)\bmod{n})$ 
cross, since assuming otherwise contradicts the definition of cyclic ladder lotteries.

We define some notations.
Let $R$ be a region enclosed by pseudolines in $L$.
Let $i$ and $j$ be two distinct elements in $[n]$
such that $\pline(L,i)$ and $\pline(L,j)$ include
the contour of $R$, respectively.
Note that any pseudoline 
appears as a subpseudoline on the contour of $R$ at most once in $L$.
Let $\plset(i,j,R)$ be the set of the pseudolines such that
they 
pass through $R$ by crossing with $\pline(L,i)$ and $\pline(L,j)$.
We denote by $\crset(i,j,R)$ the set of the intersections of the pseudolines in $\plset(i,j,R)$ inside $R$.
Then, we define a binary relation $\prec_{i,j,R}$ on $\crset(i,j,R)$, as follows.
Let $\ell$ be an element in $[n]$ such that $\pline(L,\ell) \in \plset(i,j,R)$ holds.
Let $\mseq{c_1,c_2,\ldots ,c_r}$ be the intersections on $\pline(L,\ell)$ in $\crset(i,j,R)$ from $\cross(i,\ell)$ to $\cross(j,\ell)$.
Here, note that $c_1$ and $c_r$ are adjacent to $\cross(\ell,i)$ and $\cross(\ell,j)$, respectively.
We define an order such that $c_h$ is 
\emph{smaller} than $c_{h+1}$ for each $h=1,2,\ldots ,q-1$ and denote $c_h \prec_{i,j,R} c_{h+1}$.
Then, we can observe that the above binary relation generates a partially order $\prec_{i,j,R}$ on $\crset(i,j,R)$.

Now, we give a proof of Lemma~\ref{lem:top_intersection}.

\begin{lemma}\label{lem:top_intersection}
Let $\pi$ and $\bm{y}$ be a permutation in $\frakS_n$ and an almost optimal displacement vector of $\pi$, respectively.
Let $L$ be a cyclic ladder lottery in $\llsetone(\pi,\bm{y})$.
Let $i \in [n]$ be an element such that $\pline(L,i)$ and 
$\pline(L,(i+1)\bmod{n})$ cross.
Then, there exists a reconfiguration sequence under braid relations between $L$
and $L'$, where $L'$ is a cyclic ladder lottery in $\llsetone(\pi,\bm{y})$ such that 
$\cross(i,(i+1)\bmod n)$ appears as a topmost intersection in $L'$.
\end{lemma}
\begin{proof}
In this proof, we describe all the indices in modulo $n$ and omit ``mod $n$'' for readability.
Let $R$ be the closed region of $L$ enclosed by
(1) $\pline(L,i)$,
(2) $\pline(L,i+1)$,
and
(3) the line segment connecting the top endpoints of $\pline(L,i)$ and $\pline(L,i+1)$.
If $R$ has no subpseudoline inside, 
$\cross(i,i+1)$ is a topmost intersection in $L$, and hence the claim holds.
Now, we assume otherwise.
We focus on $\crset(i,i+1,R)$ and have the following two cases.
Note that, in this situation, $\crset(i,i+1,R)$ coincides with all the intersections inside $R$.

\begin{mycase}{1}{$\crset(i,i+1,R) = \emptyset$.}
In this case, $R$ includes subpseudolines in parallel, as shown in \figurename~\ref{fig:top_intersection}(a).
We show that one can reduce 
$\msize{\plset(i,i+1,R)}$ by one.
Then, let $u$ and $v$ be two intersections adjacent to and above $\cross(i,i+1)$ on $\pline(L,i)$ and $\pline(L,i+1)$,
respectively~(see \figurename~\ref{fig:top_intersection}(a)).
Let $j$ be an element such that $\pline(L,j)$ includes $u$ and $v$.
Then, 
$\mset{i, i+1, j}$ is a minimal tangled triple.
By applying a braid relation to the triple, we obtain a cyclic ladder lottery such that $R$ includes one less subpseudoline.
\end{mycase}

\begin{mycaselast}{2}{$\crset(i,i+1,R) \neq \emptyset$.}
In this case, $R$ includes one or more intersections in its inside.
We show that one can reduce 
$\msize{\crset(i,i+1,R)}$ by one by applying a braid relation.
Note that $(\crset(i,i+1,R),\prec_{i,i+1,R})$ is a partially ordered set.
Let $u$ be a minimal intersection in $\crset(i,i+1,R)$.
Let $j,k \in [n]$ be the two distinct elements such that $\pline(L,j)$ and $\pline(L,k)$ cross on $u$, that is $u = \cross(j,k)$.
See \figurename~\ref{fig:top_intersection}(b).
Then, $\mset{i,j,k}$ is a minimal tangled triple and by applying a braid relation to $\mset{i,j,k}$,
we can reduce 
$\msize{\crset(i,i+1,R)}$ by one.
\end{mycaselast}
\begin{figure}[tb]
  \centerline{\includegraphics[width=0.6\linewidth]{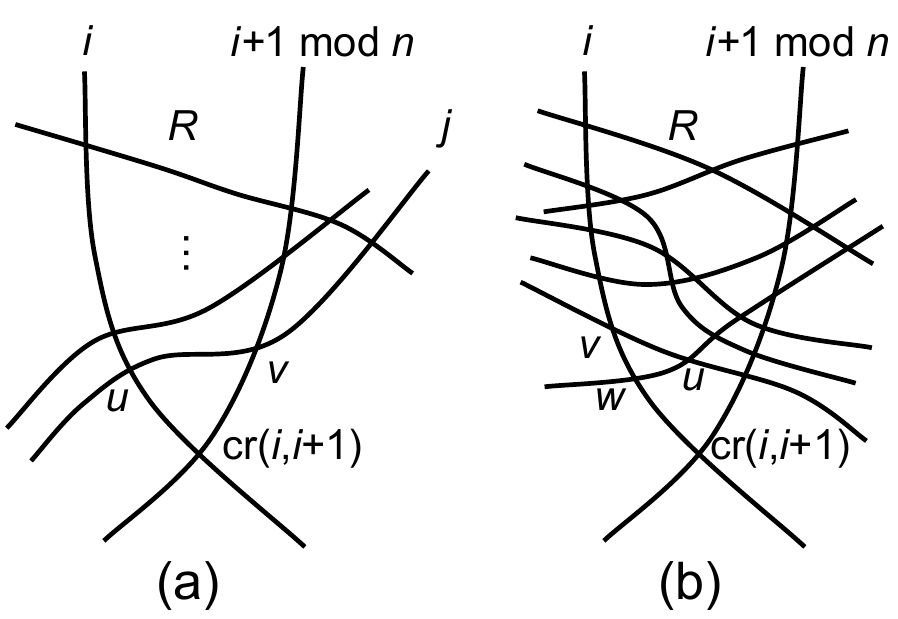}}
  \caption{Illustrations for (a) Case~1 and  (b) Case~2.}
\label{fig:top_intersection}
\end{figure}

Now, we are ready to complete the proof.
Starting from $L$, by repeatedly applying the above case analysis to $L$, we finally have 
a cyclic ladder lottery in $\llsetone(\pi,\bm{y})$ such that 
$\cross(i,i+1)$ appears as a topmost intersection.
\qed
\end{proof}

Using Lemma~\ref{lem:top_intersection}, we can prove the following theorem, which claims that any two cyclic ladder lotteries in $\llsetone(\pi,\bm{y})$ with the same displacement vector are always reachable.

\begin{theorem}
\label{thm:reconfCLL-DV}
Let $\pi$ be a permutation in $\frakS_n$.
Let $L,L'$ be two cyclic ladder lotteries in $\llsetone(\pi)$.
Then, $\DV(L) = \DV(L')$ holds if and only if $L$ and $L'$ are reachable under braid relations.
\end{theorem}

\begin{proof}
($\to$) Proof by induction on the number of intersections in $L$ (and $L'$).
When the number of intersections
in $L$ is 0, then $L = L'$.
Hence, the claim holds.

Suppose that the number of intersections is more than 0 in $L$.
We describe all the indices in modulo $n$ and omit ``mod $n$'' for readability.
Let $x$ be a topmost intersection in $L$.
If ties exist, we choose any intersection as $x$.
Let $i$ be the element in $[n]$ such that $x = \cross(i,i+1)$ holds in $L$.
Let $x'$ be the intersection of $\pline(L',i)$ and $\pline(L',i+1)$ in $L'$, that is $x' = \cross(i,i+1)$ in $L'$.
If $x'$ is not a topmost intersection in $L'$,
from Lemma~\ref{lem:top_intersection}, 
we can construct a cyclic ladder lottery in $\llset(\pi,\DV(L'))$ such that $x'$ is a topmost intersection from $L'$
by repeatedly applying braid relations.
We denote the obtained ladder lottery by $L''$.
By removing $x$ and $x'$ in $L$ and $L''$, respectively, we obtain two 
cyclic ladder lotteries $M$ and $M'$ with one less intersection of a permutation $\pi \circ \tau_{i,i+1}$,
where $\tau_{i,i+1}$ is the transposition of $i$ and $i+1$, and both $M$ and $M'$ has the same 
displacement vector.
From the induction hypothesis,
$M$ and $M'$ are reachable under braid relations.

($\leftarrow$) 
We prove the claim by contraposition.
Assume that $\DV(L) \neq \DV(L')$ holds.
Since applying a braid relation does not change a displacement vector of a cyclic ladder lottery, there is no reconfiguration sequence between $L$ and $L'$ under braid relations.
\qed
\end{proof}

The following corollary is immediately from  Theorem~\ref{thm:reconfCLL-DV} and the proofs of Lemma~\ref{lem:top_intersection} and Theorem~\ref{thm:reconfCLL-DV}.

\begin{corollary}\label{cor:reachable}
For any instance of \reconfCLLDV, the answer 
is yes and one can construct a reconfiguration sequence.
\end{corollary}

\subsection{Shortest reconfiguration}

In Section~\ref{sec:reach}, we showed that any two cyclic ladder lottereis in $\llsetopt(\pi,\bm{x})$ are always reachable under braid relations.
In this subsection, 
we investigate shortest reconfiguration sequences between two cyclic ladder lottereis in $\llsetopt(\pi,\bm{x})$.

Similar to Section~\ref{sec:reach}, we also investigate shortest reconfigurations for the cyclic ladder lotteries in $\llsetone(\pi,\bm{y})$, where $\bm{y}$ is an almost optimal displacement vector of $\pi$.
We denote by $\lefttangled(L)$ the set of the left tangled triples in a cyclic ladder lottery $L$.

\begin{lemma}\label{lem:lefttangled}
Let $\pi \in \frakS_n$ 
and let $\bm{y}$ be an almost optimal displacement vector of $\pi$.
For two cyclic ladder lotteries $L,L' \in \llsetone(\pi,\bm{y})$,
if $\lefttangled(L) = \lefttangled(L')$, $L=L'$ holds.
\end{lemma}

\begin{proof}
We prove the claim by contraposition.
Since $L \neq L'$, there exists an element $i \in [n]$ such that,
for two intersections $\cross(i,j)$ and $\cross(i,k)$ ($j,k \in [n]$),
$\cross(i,j)$ is closer to the top endpoint of $\pline(L,i)$ than $\cross(i,k)$ in $L$ and
$\cross(i,k)$ is closer to the top endpoint of $\pline(L',i)$ than $\cross(i,j)$ in $L'$.
From Theorem~\ref{thm:reconfCLL-DV}, $L$ and $L'$ are reachable under braid relations.
Hence, the two pseudolines of $j$ and $k$ cross in both $L$ and $L'$.
Therefore, the triple $\mset{i,j,k}$ is tangled in both $L$ and $L'$.
Moreover, if $\mset{i,j,k}$ is a 
left (resp.\ right) tangled triple in $L$, it is a right (resp.\ left) tangled triple in $L'$.
Hence, $\lefttangled(L)\neq \lefttangled(L')$ holds.
\qed
\end{proof}

We denote the length of a shortest reconfiguration sequence between $L, L' \in \llsetone(\pi,\bm{x})$ by $\optclldv(L,L')$.
$\symdiff{\lefttangled(L)}{\lefttangled(L')}$ denotes the symmetric difference of $\lefttangled(L)$ and $\lefttangled(L')$.
Then, we have the following lemma.

\begin{lemma}\label{lem:lower}
Let $L,L'$ be two cyclic ladder lotteries in $\llsetone(\pi,\bm{y})$, where $\pi \in \frakS_n$ and $\bm{y}$ is an almost displacement vector of $\pi$.
Then, 
\[
\optclldv(L,L') \geq \msize{\symdiff{\lefttangled(L)}{\lefttangled(L')}}
\]
holds.
\end{lemma}

\begin{proof}
    From Lemma~\ref{lem:lefttangled}, we can observe that $\symdiff{\lefttangled(L)}{\lefttangled(L')} = \emptyset$ if and only if $L=L'$.
    Let $M$ be the cyclic ladder lottery obtained from $L$ by applying a braid relation to any minimal tangled triple. Then, $\msize{\symdiff{\lefttangled(L)}{\lefttangled(L')}}-1 \leq \msize{\symdiff{\lefttangled(L)}{\lefttangled(M)}}$ holds.
    Therefore, the claim holds.
    \qed
\end{proof}

We show an upper bound of $\optclldv(L,L')$, 
as stated in the following lemma.

\begin{lemma}\label{lem:min_exist}
    Let $L,L'$ be two distinct cyclic ladder lotteries in $\llsetone(\pi,\bm{y})$, where $\pi \in \frakS_n$ and $\bm{y}$ is an almost optimal displacement vector of $\pi$.
    Then, there exists a minimal tangled triple in $\symdiff{\lefttangled(L)}{\lefttangled(L')}$.
\end{lemma}
\begin{proof}
    We prove by induction on $n$.
    If $n=3$, we have $\msize{\symdiff{\lefttangled(L)}{\lefttangled(L')}} = 1$ since $L$ and $L'$ are distinct.
    Then, it can be observed that the tangled triple in $\symdiff{\lefttangled(L)}{\lefttangled(L')}$ is minimal.

    We assume that the claim holds for $n-1$.
    Let $K$ and $K'$ be the two ladder lotteries obtained from $L$ and $L'$ by removing $\pline(n,L)$ and $\pline(n,L')$, respectively.
    If $K$ and $K'$ are equivalent, then one can find a minimal tangled triple in $\symdiff{\lefttangled(L)}{\lefttangled(L')}$, as stated in the following claim.

    \begin{claim}\label{claim:equiv}
        If $K$ and $K'$ are equivalent, there exists a minimal tangled triple in $\symdiff{\lefttangled(L)}{\lefttangled(L')}$.
    \end{claim}
    \begin{proof}
        Since $K$ and $K'$ are equivalent, any tangled triple in $\symdiff{\lefttangled(L)}{\lefttangled(L')}$ includes $n$.
        Let $\mseq{c_1,c_2,\ldots ,c_r}$ be the sequence of the intersections on $\pline(L,n)$ from its top endpoint to bottom endpoint.
        Since $L$ and $L'$ are distinct, there exists a pair $(c_i,c_j)$ ($i<j$), where $c_i=\cross(n,x)$ and $c_j=\cross(n,y)$, such that $\mset{n,x,y} \in \symdiff{\lefttangled(L)}{\lefttangled(L')}$. 
        Now, among such pairs, we choose the smallest index $i$, where for distinct two pairs $(c_s,c_t)$ and $(c_{s'},c_{t'})$ ($s< t$ and $s'<t'$), we say that $(c_s,c_t)$ is smaller than $(c_{s'},c_{t'})$ if (1) $s<s'$ or (2) $s=s'$ and $t<t'$ hold.
        
        If $\mset{n,x,y}$ is minimal, we are done.
        Thus, we assume otherwise that it is not minimal.
        Let $\cross(x,k)$ be the intersection between $c_i = \cross(x,n)$ and $\cross(x,y)$ such that $\cross(x,k)$ is adjacent to $c_i$.\footnote{Note that $k=y$ holds if there is no intersection between $c_i$ and $\cross(x,y)$.}
        \figurename s~\ref{fig:configurations}(a)--(e) shows the candidates of the possible cyclic ladder lotteries induced by only the four pseudolines $\pline(L,x)$, $\pline(L,y)$, $\pline(L,k)$, and $\pline(L,n)$.
        By the case analysis below, we can observe that \figurename s~\ref{fig:configurations}(a)--(d) do not occur.
        \begin{itemize}
        \item[(a)] The order of intersections on $\pline(L,k)$ is $\cross(n,k)$, $\cross(x,k)$, and $\cross(y,k)$.
        See \figurename~\ref{fig:configurations}(a).
        In this case, $\mset{k,n,x} \in \symdiff{\lefttangled(L)}{\lefttangled(L')}$, which contradicts the choice of $i$.
        \item[(b)] The order of intersections on $\pline(L,k)$ is $\cross(L,x)$, $\cross(L,n)$. See \figurename~\ref{fig:configurations}(b).
        In this case, $\mset{n,k,x} \in \symdiff{\lefttangled(L)}{\lefttangled(L')}$ also holds. This contradicts the choice of $i$.
        \item[(c)] The order of intersections on $\pline(L,k)$ is $\pline(L,x)$ and $\pline(L,y)$. See \figurename~\ref{fig:configurations}(c).
        It can be observed that there is no reconfiguration sequence between $L$ and $L'$, since $K$ and $K'$ are equivalent. This contradicts Theorem~\ref{thm:reconfCLL-DV}.
        \item[(d)] The order of the intersections on $\pline(L,k)$ is $\pline(L,y)$, $\pline(L,x)$, and $\pline(L,n)$.
        See \figurename~\ref{fig:configurations}(d).
        In this case, $\mset{n,k,x} \in \symdiff{\lefttangled(L)}{\lefttangled(L')}$ holds, which contradicts the choice of $i$.
        \end{itemize}
        \begin{figure}[tb]
        \centerline{\includegraphics[width=1.0\linewidth]{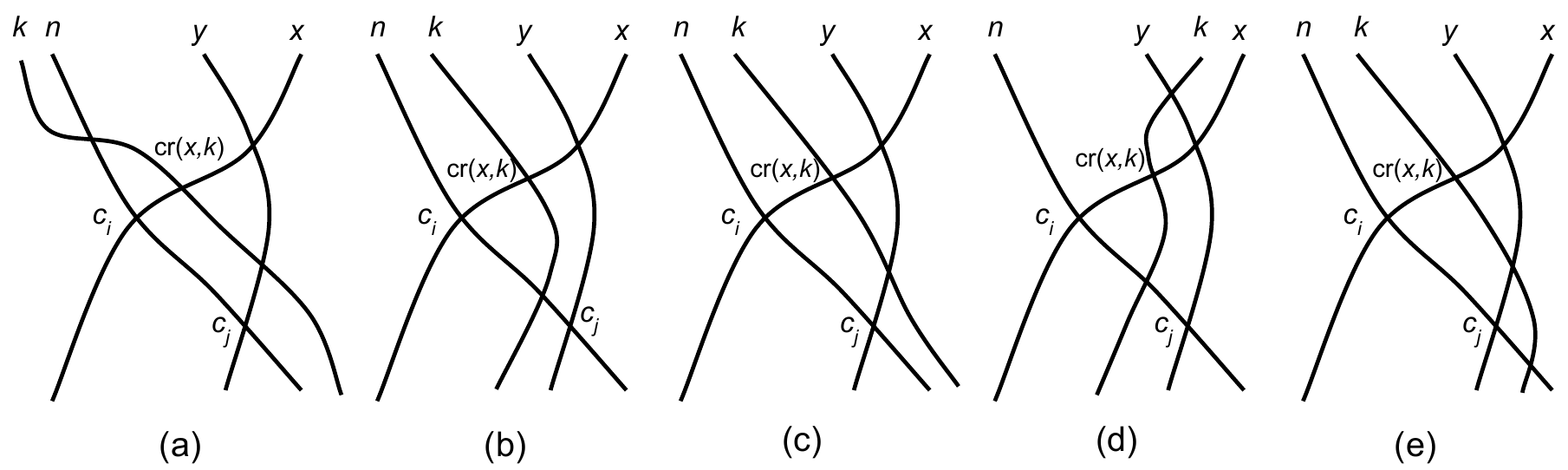}}
        \caption{Illustrations for the candidates of possible cyclic ladder lotteries induced by four pseudolines.}
        \label{fig:configurations}
    \end{figure}    
        Therefore, below, we consider only the cyclic ladder lotteries illustrated in \figurename~\ref{fig:configurations}(e).

    \begin{figure}[tb]    \centerline{\includegraphics[width=0.5\linewidth]{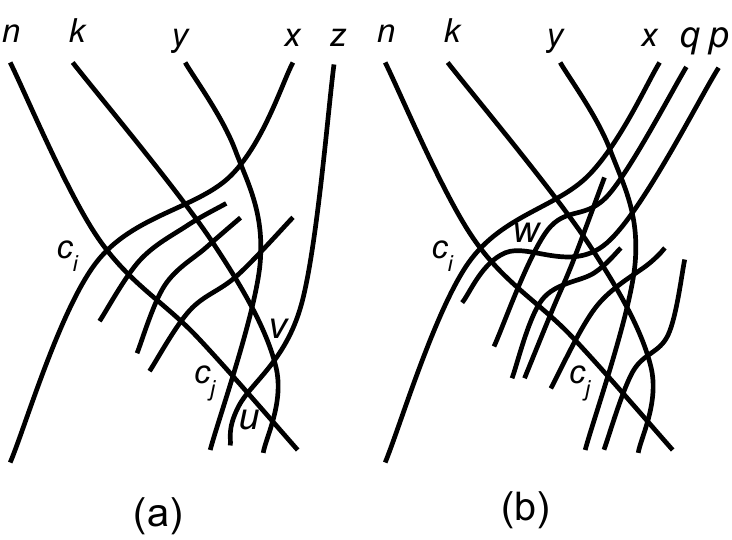}}
        \caption{Illustrations for the proof of Claim~\ref{claim:equiv}. (a) The pseudolines of $\pline(L,n)$, $\pline(L,x)$, $\pline(L,y)$, and $\pline(L,k)$. (b) The case of $\crset(n,k,R) = \emptyset$. (c) The case of $\crset(n,k,R) \neq \emptyset$.}
        \label{fig:equiv}
    \end{figure}
    Then, there exists the intersection $\cross(n,k)$ and $\mset{n,k,x}$ is in $\symdiff{\lefttangled(L)}{\lefttangled(L')}$, since $L$ and $L'$ are reachable from Theorem~\ref{thm:reconfCLL-DV} and $\mset{n,x,y}$ is in $\symdiff{\lefttangled(L)}{\lefttangled(L')}$.
    Also, $\cross(n,k)$ is closer to the bottom endpoint of $\pline(L,n)$ than $c_i$ from the choice of $c_i$.

    Let $R$ be the region enclosed by $\pline(L,n)$, $\pline(L,x)$, and $\pline(L,k)$.
    Note that there is no intersection between $c_i$ and $\cross(k,x)$ from the choice of $k$. 
    Therefore, we have a similar discussion of the proof of Lemma~\ref{lem:top_intersection}, as follows.
    
    First, we assume that $\crset(n,k,R) = \emptyset$ holds. 
    In this case, $R$ includes subpseudolines in parallel.
    See \figurename~\ref{fig:equiv}(a).
    Let $u$ and $v$ be two intersections adjacent to and above $\cross(n,k)$ on $\pline(L,n)$ and $\pline(L,k)$,
respectively.
    Let $z \in [n]$ be an element such that $\pline(L,z)$ includes $u$ and $v$.
    Then, $\mset{n,k,z}$ is minimal and in $\symdiff{\lefttangled(L)}{\lefttangled(L')}$.

    Next, we assume that $\crset(n,k,R) \neq \emptyset$ holds.
    In this case, $R$ includes one or more intersections inside.
    See \figurename~\ref{fig:equiv}(b).
    Note that $(\crset(n,k,R),\prec_{n,k,R})$ is a partially ordered set.
    Let $u$ be a minimal intersection in $\crset(n,k,R)$.
    Let $p,q \in [n]$ be the two distinct elements such that $w=\cross(p,q)$ holds.
    Then, $\mset{n,p,q}$ is minimal and in $\symdiff{\lefttangled(L)}{\lefttangled(L')}$.
    \qed    
    \end{proof}

    From Claim~\ref{claim:equiv}, the claim holds if $K$ and $K'$ are equivalent.
    Now we assume that $K$ and $K'$ are distinct.
    From the hypothesis, there exists a minimal tangled triple $\mset{i,j,k}$ in $\symdiff{\lefttangled(K)}{\lefttangled(K')}$.
    We assume that $\mset{i,j,k} \in \lefttangled(K)$ holds.
    If $\mset{i,j,k}$ is also minimal in either $L$ or $L'$, $\mset{i,j,k}$ is a target triple and we are done.
    Thus, we assume that $\mset{i,j,k}$ is non-minimal in $L$ and $L'$.
    Let $R_L$ (resp.\ $R_{L'}$) be the closed region enclosed by $\pline(L,i)$, $\pline(L,j)$, and $\pline(L,k)$ (resp.\ $\pline(L',i)$, $\pline(L',j)$, and $\pline(L',k)$).
    Then, $\pline(n,L)$ and $\pline(n,L')$ cross the region $R_L$ and $R_{L'}$, respectively.
    We have the following case analysis.
    
    \begin{mycase}{1}{$\pline(L,n)$ crosses with $\pline(L,j)$ when $\pline(L,n)$ enters $R_L$.}
    In this case, if $\pline(L,n)$ has no intersection with $\pline(L,k)$~(see \figurename~\ref{fig:shortest_minimal1}(a)), the triple $\mset{i,j,k}$ is minimal in $L'$~(see \figurename~\ref{fig:shortest_minimal1}(b)). Hence, we assume that there exists the intersection $\cross(k,n)$ in both $L$ and $L'$.
    Then, $R_L$ has the following two cases
    
    \begin{mycase}{1-1}{$\pline(L,n)$ crosses with $\pline(L,i)$ when $\pline(L,n)$ leaves $R_L$.}
        In this case, $R_L$ has the two situations, illustrated in \figurename s~\ref{fig:shortest_minimal1}(c) and (d).
        For the case of \figurename~\ref{fig:shortest_minimal1}(c), $R_{L'}$ has two situations, as illustrated in \figurename s~\ref{fig:shortest_minimal1}(e) and (f). Then, $\mset{i,j,n}$ is minimal in $L$ for \figurename s~\ref{fig:shortest_minimal1}(e) and (f).
        For the case of \figurename~\ref{fig:shortest_minimal1}(d), $R_{L'}$ has only a situation, as illustrated in \figurename~\ref{fig:shortest_minimal1}(g).
        Then, $\mset{i,j,n}$ is also minimal in $L$.
    \end{mycase}
    
    \begin{mycase}{1-2}{$\pline(L,n)$ crosses with $\pline(L,k)$ when $\pline(L,n)$ leaves $R_L$~(\figurename~\ref{fig:shortest_minimal1}(h)).}
        Then, $R_{L'}$ have two situations, as illustrated in \figurename s~\ref{fig:shortest_minimal1}(e) and (f). 
        Then, $\mset{j,k,n}$ is minimal in $L$ for the two cases.
    \end{mycase}
    \end{mycase}
    \begin{figure}[tb]
    \centerline{\includegraphics[width=1.0\linewidth]{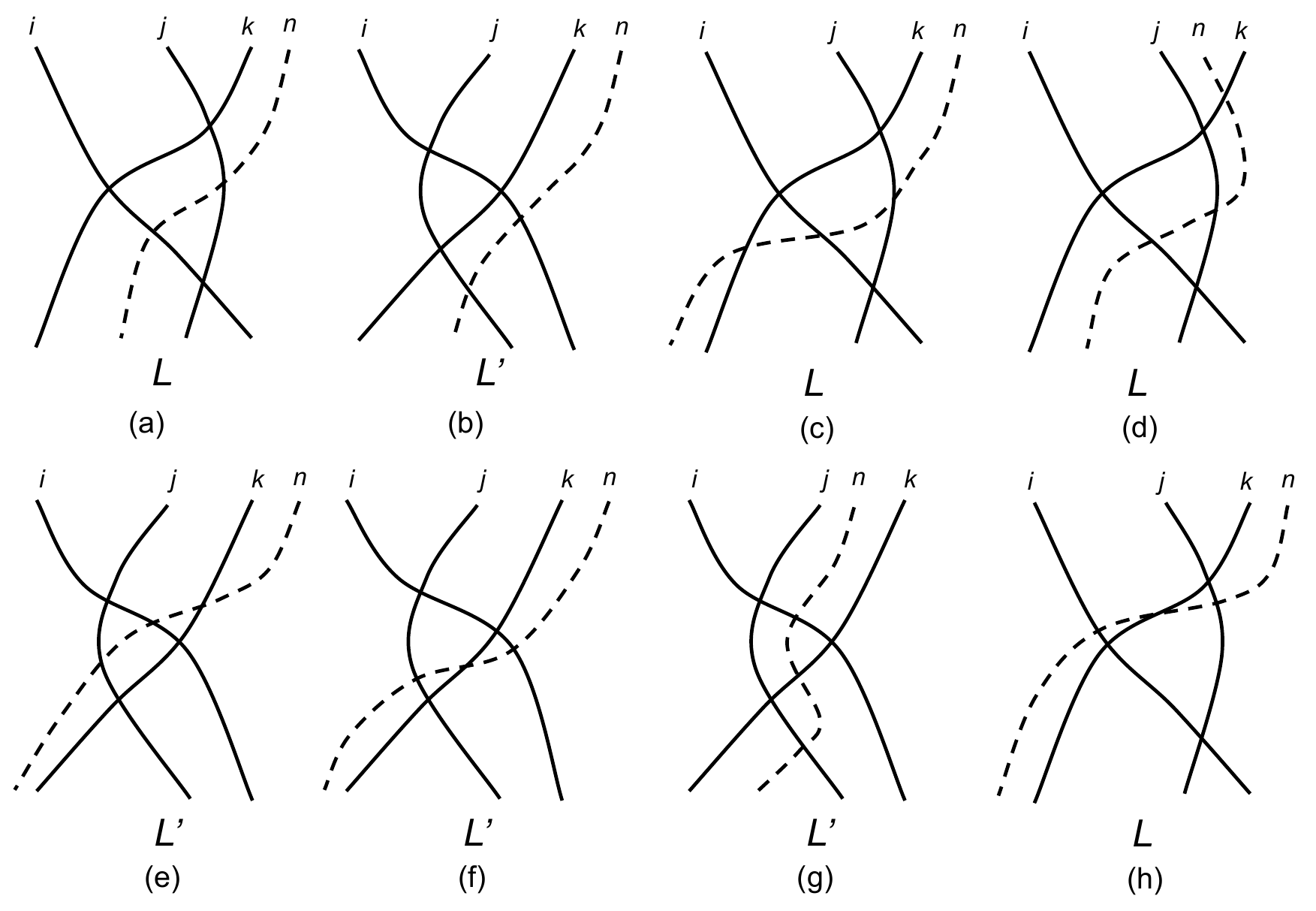}}
    \caption{Illustrations for Case~1.}
    \label{fig:shortest_minimal1}
    \end{figure}

    \begin{mycaselast}{2}{$\pline(L,n)$ crosses with $\pline(L,k)$ when $\pline(L,n)$ enters $R_L$.}
    If $\pline(L,n)$ does not cross with $\pline(L,i)$~(\figurename~\ref{fig:shortest_minimal2}(a)), 
    $\mset{n,j,k}$ is a minimal tangled triple in $L$ and right tangle triple in $L'$~(\figurename~\ref{fig:shortest_minimal2}(b)).
    Similarly, if $\pline(L,n)$ does not cross with $\pline(L,j)$~(\figurename~\ref{fig:shortest_minimal2}(b)), 
    then $\mset{i,n,k}$ is a minimal tangled triple in $L$ and a right tangled triple in $L'$~(\figurename~\ref{fig:shortest_minimal2}(d)).
    Hence, in this case, we can assume that $\pline(L,n)$ crosses with all of $\pline(L,i)$, $\pline(L,j)$, and $\pline(L,k)$.
    Then, $R_L$ has the following two subcases.
    
    %
    %
    \begin{mycase}{2-1}{$\pline(L,n)$ crosses with $\pline(L,j)$ when $\pline(L,n)$ leaves $R_L$.}
    In this subcase, $R_L$ has two situations, as illustrated in \figurename~\ref{fig:shortest_minimal3}(a) and (b). 
    For the case of \figurename~\ref{fig:shortest_minimal3}(a), $R_{L'}$ has two situations, as illustrated in \figurename~\ref{fig:shortest_minimal3}(c) and (d).
    Then, $\mset{n,j,k}$ is minimal in $L$.
    For the case of \figurename~\ref{fig:shortest_minimal3}(b), of $R_{L'}$ has two situations, as illustrated in \figurename~\ref{fig:shortest_minimal3}(e) and (f).
    $\mset{n,j,k}$ is also a target triple.
    \end{mycase}
    
    %
    %
    \begin{mycase}{2-2}{$\pline(L,n)$ crosses with $\pline(L,i)$ when $\pline(L,n)$ leaves $R_L$.}
    In this subcase, we have the two situations of $R_L$, as illustrated in \figurename~\ref{fig:shortest_minimal3}(g) and (h). 
    For the case of \figurename~\ref{fig:shortest_minimal3}(g), $R_{L'}$ has two situations, as illustrated in \figurename~\ref{fig:shortest_minimal3}(c) and (d).
    Then, $\mset{i,n,k}$ is minimal in $L$.
    For the case of \figurename~\ref{fig:shortest_minimal3}(h),
    $R_{L'}$ has a situation, as illustrated in \figurename~\ref{fig:shortest_minimal3}(i).
    Then, $\mset{i,n,k}$ is also minimal in $L$.
    \end{mycase}
    \end{mycaselast}
    \begin{figure}[tb]
    \centerline{\includegraphics[width=0.83\linewidth]{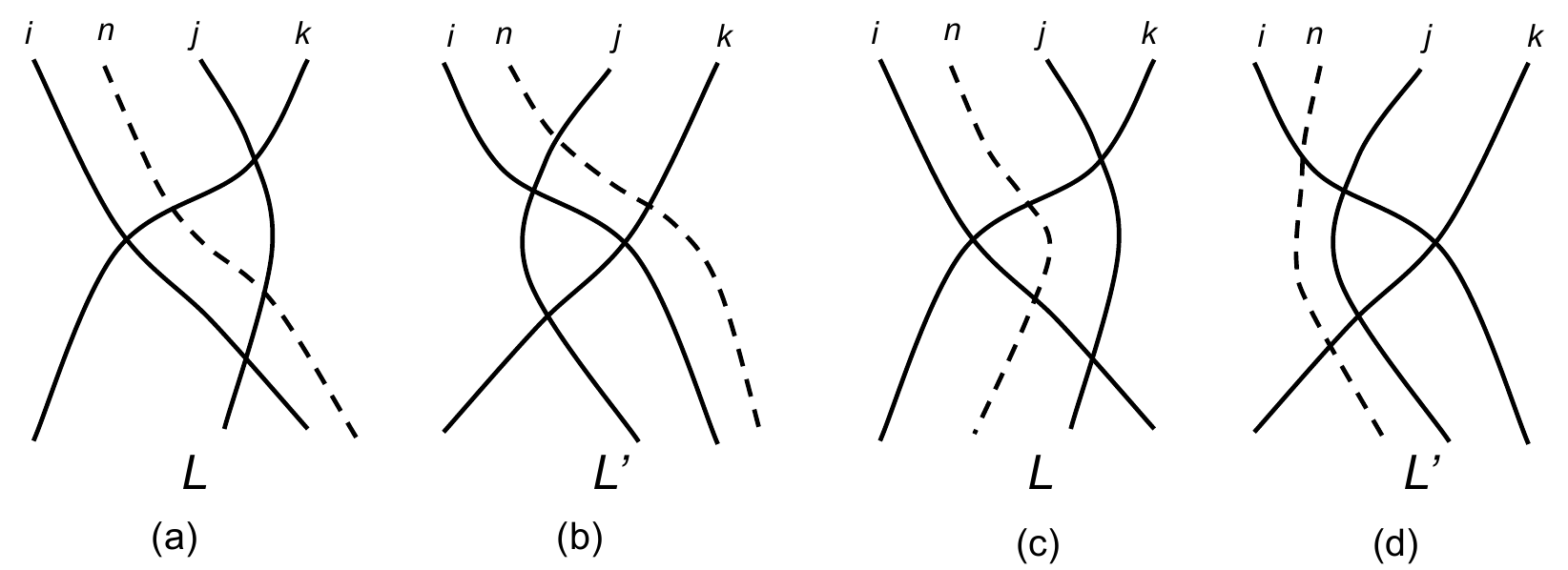}}
    \caption{Illustrations for Case~2.}
    \label{fig:shortest_minimal2}
    \end{figure}
    \begin{figure}[tb]
    \centerline{\includegraphics[width=0.83\linewidth]{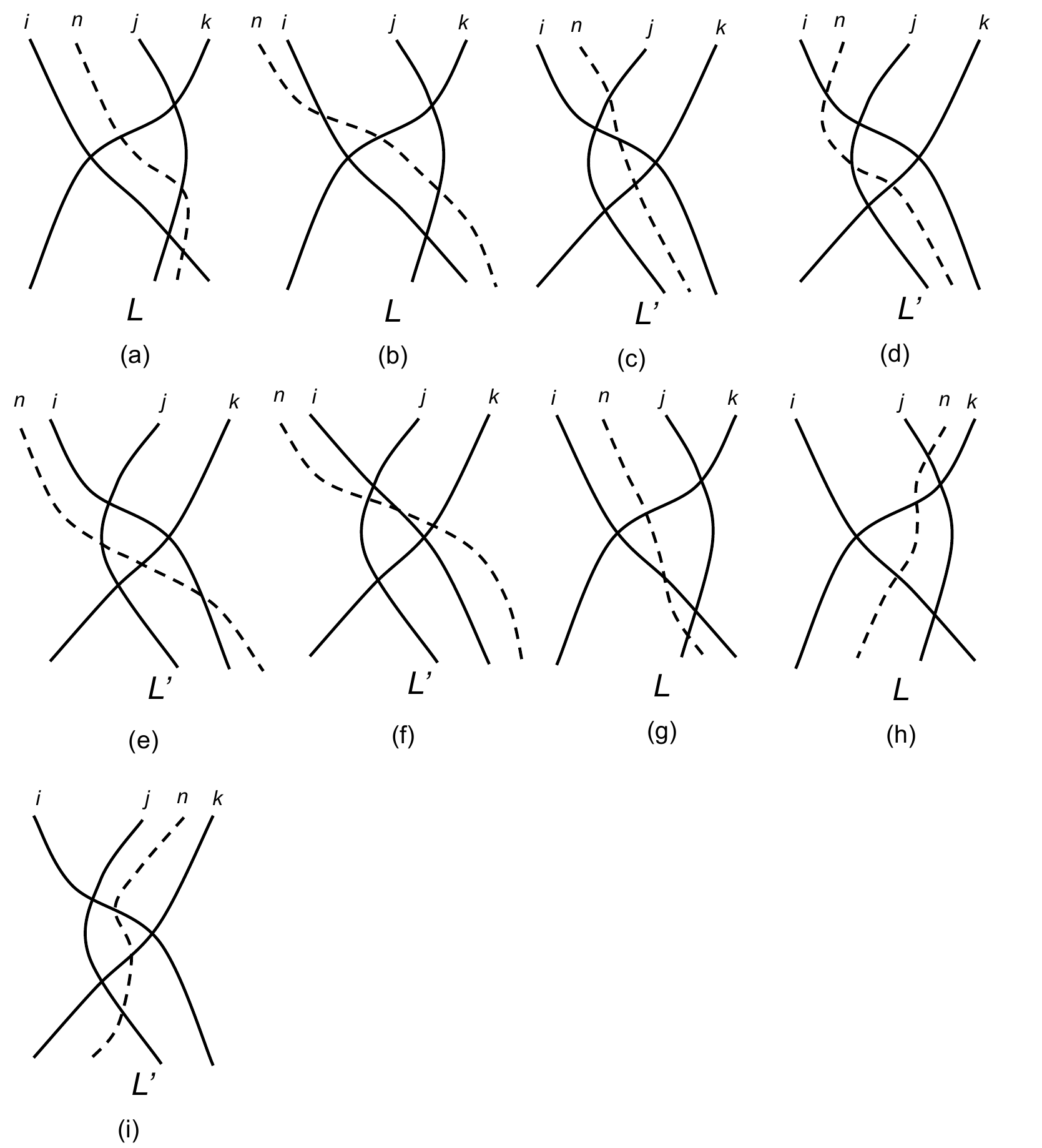}}
    \caption{Illustrations for Cases~(2-1) and (2-2).}
    \label{fig:shortest_minimal3}
    \end{figure}

    In every case above, there exists a minimal tangled triple in $\symdiff{\lefttangled(L)}{\lefttangled(L')}$.
    We have the same case analysis for the case that $\mset{i,j,k} \in \lefttangled(L')$.
    Therefore, the claim holds.
    \qed
\end{proof}

Now, we are ready to show the main theorem of this subsection.

\begin{theorem}
Let $L,L'$ be two distinct cyclic ladder lotteries in $\llsetone(\pi,\bm{y})$, where $\pi\in \frakS_n$ and $\bm{y}$ is an almost optimal displacement vector of $\pi$.
    Then, 
    $$
    \optclldv(L,L') =\msize{\symdiff{\lefttangled(L)}{\lefttangled(L')}}.
    $$
\end{theorem}
\begin{proof}
    From Lemma~\ref{lem:min_exist}, one can find a minimal tangled triple $\symdiff{\lefttangled(L)}{\lefttangled(L')}$ if $L$ and $L'$ are distinct.
    By repeatedly applying a braid relation to a minimal tangled triple, we have a reconfiguration sequence between $L$ and $L'$. Note that the length of the reconfiguration sequence is at most $\msize{\symdiff{\lefttangled(L)}{\lefttangled(L')}}$.
    Moreover, from Lemma~\ref{lem:lower}, the claim holds.
    \qed
\end{proof}

By abuse of notation, we denote the length of a shortest reconfiguration sequence between $L, L' \in \llsetopt(\pi,\bm{x})$ by $\optclldv(L,L')$.
The following corollary is immediate from the theorem.

\begin{corollary}
Let $L,L'$ be two distinct cyclic ladder lotteries in $\llsetopt(\pi,\bm{x})$, where $\pi\in \frakS_n$ and $\bm{x}$ is an almost displacement vector of $\pi$.
    Then, 
    $$
    \optclldv(L,L') =\msize{\symdiff{\lefttangled(L)}{\lefttangled(L')}}.
    $$
\end{corollary}

\subsection{Enumeration}
In this subsection, we consider the problem of enumerating all the optimal cyclic ladder lotteries in $\llsetopt(\pi, \bm{x})$.
The formal description of the problem is as follows.

\medskip
\noindent
\textbf{Problem:} Enumeration of optimal cyclic ladder lotteries with optimal displacement vector~(\enumCLLDV) \\ 
\textbf{Instance:} 
A permutation $\pi$ and an optimal displacement vector $\bm{x}$ of $\pi$.\\
\textbf{Output:}
All the cyclic ladder lotteries in $\llsetopt(\pi, \bm{x})$ without duplication.
\medskip

As in the previous subsection,
we consider the enumeration problem for $\llsetone(\pi,\bm{y})$, where $\bm{y}$ is an almost optimal displacement vector of $\pi$ and propose an enumeration algorithm for $\llsetone(\pi,\bm{y})$,
since the algorithm can be applied to \enumCLLDV.
From Theorem~\ref{thm:reconfCLL-DV},
the reconfiguration graph of $\llsetone(\pi,\bm{y})$ under braid relations is connected.
This implies that the reverse search~\cite{AvisF96} can be applied to enumerate them.
In this subsection, we propose an enumeration algorithm for $\llsetone(\pi,\bm{y})$ using the reverse search.
In the reverse search, we first define a rooted tree structure on a set of enumerating objects.
By traversing the tree structure,
we enumerate all the enumerating objects.
We first define the root object and define a rooted tree structure on $\llsetone(\pi,\bm{y})$.
Then, we describe the algorithm that traverses the tree structure.

Let $\pi \in \frakS_n$ 
and let $\bm{y}$ be an almost optimal displacement vector of $\pi$.
A cyclic ladder lottery $L$ in $\llsetone(\pi,\bm{y})$ is 
a \emph{root} of $\llsetone(\pi, \bm{y})$
if $\lefttangled(L) = \emptyset$ holds.
If a cyclic ladder lottery in $\llsetone(\pi, \bm{y})$ has no tangled triple,
then $\llsetone(\pi, \bm{y})$ contains only one ladder lottery.
For convenience, in the case, we define the ladder lottery as a root.
From Lemma~\ref{lem:lefttangled}, we have the following corollary, which states the uniqueness of a root of $\llsetone(\pi,\bm{y})$.

\begin{corollary}
\label{cor:root_is_unique}
Let $\pi \in \frakS_n$ 
and let $\bm{y}$ be an almost optimal displacement vector of $\pi$.
Then, the root $L_0$ of $\llsetone(\pi, \bm{y})$ is unique.
\end{corollary}

One can show that any cyclic ladder lottery except the root has a minimal left tangled triple, as stated in the following corollary.

\begin{corollary}\label{cor:minimal_left}
Let $\pi \in \frakS_n$ and let $\bm{y}$ be an almost optimal displacement vector of $\pi$.
Suppose that a cyclic ladder lottery in $\llsetone(\pi,\bm{y})$ contains one or more tangled triples.
Then, any 
$L \in \mset{L \in \llsetone(\pi,\bm{y}) \mid \lefttangled(L) \neq \emptyset}$
has a minimal left tangled triple.
\end{corollary}

\begin{proof}
    In the statement of Lemma~\ref{lem:min_exist},
    assume that $L'$ is a root $L_0$ of $\llsetone(\pi,\bm{y})$.
    Then, since $\lefttangled(L_0) = \emptyset$, the statement of this corollary holds.
\qed
\end{proof}

Let $L$ be a cyclic ladder lottery in $\llsetone(\pi,\bm{y})$.
Let $\mset{i,j,k}$ and $\mset{i',j',k'}$ be two distinct left tangled triples in $L$,
and suppose that $i<j<k$ and $i'<j'<k'$ hold.
We say that $\mset{i,j,k}$ is \emph{smaller} than  $\mset{i',j',k'}$
if either $i < i'$ holds, $i=i'$ and $j<j'$ hold, or $i=i'$, $j=j'$, and $k<k'$ hold.
The \emph{parent}, denoted by $\parent(L)$, of 
$L \in \mset{L \in \llsetone(\pi,\bm{y}) \mid \lefttangled(L) \neq \emptyset}$
is the cyclic ladder lottery obtained from $L$ by applying a braid relation to the smallest minimal left tangled triple in $L$.
We say that $L$ is a \emph{child} of $\parent(L)$.
$L \in \mset{L \in \llsetone(\pi,\bm{y}) \mid \lefttangled(L) \neq \emptyset}$.
Moreover, the parent is unique from its definition.
\begin{lemma}\label{lem:toroot}
Let $\pi \in \frakS_n$ and let $\bm{y}$ be an almost optimal displacement vector of $\pi$.
Let $L$ be a cyclic ladder lottery in $\llsetone(\pi,\bm{y})$.
By repeatedly finding the parent from $L$,
we have the root $L_0$ of $\llsetone(\pi,\bm{y})$.
\end{lemma}
%
%
\begin{proof}
If $\lefttangled(L) = \emptyset$ hold, $L$ is the root.
Hence, we assume otherwise that $\lefttangled(L) \neq \emptyset$ holds.
From Corollary~\ref{cor:minimal_left},
$L$ has one or more minimal left tangled triples.
Let $\mset{i,j,k}$ be the smallest one among them.
By applying a braid relation to $\mset{i,j,k}$,
we have 
the parent $\parent(L)$ of $L$.
Note that $\msize{\lefttangled(\parent(L))} = \msize{\lefttangled(L)} -1$ holds.
Here, it is easy to observe that $\msize{\lefttangled(L_0)} = 0$ for any root $L_0$ of $\llsetone(\pi,\bm{y})$ and $\msize{\lefttangled(M)} > 0$ for any $M \in \mset{L \in \llsetone(\pi,\bm{x}) \mid \lefttangled(L) \neq \emptyset}$.
Hence, the claim holds.
\qed
\end{proof}

Lemma~\ref{lem:toroot} implies that there always exists the root of $\llsetone(\pi,\bm{y})$.
Hence, by repeatedly finding the parent from $L$, we finally obtain the root of $\llsetone(\pi,\bm{y})$.
The \emph{parent sequence} of $L \in \llsetone(\pi,\bm{y})$ is the sequence $\mseq{L_1,L_2,\ldots ,L_p}$ such that 
\begin{listing}{aaa}
\item[(1)] $L_1$ is $L$ itself,
\item[(2)] $L_i = \parent(L_{i-1})$ for $i=2,3,\ldots ,p$, and 
\item [(3)] $L_p$ is the root $L_0$ of $\llsetone(\pi,\bm{y})$.
\end{listing}
Note that the parent sequence of the root is $\mseq{L_0}$.
The \emph{family tree} of $\llsetone(\pi,\bm{y})$ is the tree structure obtained by merging the parent sequences of all the cyclic ladder lotteries in $\llsetone(\pi,\bm{y})$.
In the family tree of $\llsetone(\pi,\bm{y})$,
the root node is the root $L_0$ of $\llsetone(\pi,\bm{y})$, each node is a cyclic ladder lottery in $\llsetone(\pi,\bm{y})$, and each edge is a parent-child relationship of two ladder lotteries in $\llsetone(\pi,\bm{y})$.

Now, we design an enumeration algorithm of all the cyclic ladder lotteries in $\llsetone(\pi,\bm{y})$.
The algorithm enumerates them by traversing the family tree of $\llsetone(\pi,\bm{y})$ starting from the root $L_0$.
To traverse the family tree, we design the following two algorithms: (1) an algorithm that constructs the root $L_0$ of $\llsetone(\pi,\bm{y})$
and (2) an algorithm that enumerates all the children of a given cyclic ladder lottery in $\llsetone(\pi,\bm{y})$.
Note that, if we have the above two algorithms,
starting from the root, we can traverse the family tree by recursively applying the child-enumeration algorithm.

The outline of how to construct the root is as follows.
First, we construct a cyclic ladder lottery from $\pi$ and $\bm{y}$, which may not be the root in $\llsetone(\pi,\bm{y})$.
Next, from the constructed cyclic ladder lottery,
by repeatedly finding parents, we obtain the root.

\begin{lemma}\label{lem:root_construction}
Let $\pi$ and $\bm{y}$ be a permutation in $\frakS_n$ and an almost optimal displacement vector of $\pi$, respectively.
One can construct the root $L_0$ of $\llset(\pi,\bm{y})$ in $\order{n+ (\inv(\bm{y}))^3}$ time.
\end{lemma}
\begin{proof}
We explain how to construct a cyclic ladder lottery from given $\pi=(\pi_1,\pi_2,\ldots ,\pi_n)$ and $\bm{y}=(y_1,y_2,\ldots ,y_n)$.
Here, we describe all the indices in modulo $n$ and omit ``mod $n$'' for readability.
Let $i$ be an element in $[n]$.
$\pline(L,i)$ and $\pline(L,i+1)$ cross if $y_i > y_{i+1}$ holds.
According to the observation,
an algorithm repeats to insert an intersection and update the current permutation and displacement vector.
First, we set $\pi^{0} = (1,2,\ldots ,n)$ and $\bm{y}^0 = \bm{y}$.
We repeat the following process.

\vspace{2mm}
\begin{listing}{aaa}
\item[Step~1:] Let $\pi^{k} = (\pi_1^{k},\pi_2^{k}, \ldots ,\pi_n^{k})$
and $\bm{y}^{k} = (y_1^{k},y_2^{k}, \ldots ,y_n^{k})$.
\item[Step~2:] Let $i$ be an index such that $\pi_i^{k}$ crosses with $\pi_{i+1}^{k}$ and create the intersection of the two pseudolines of $i$ and $i+1$.
\item[Step~3:] Update the permutation and displacement vector so that $$\pi^{k+1} = (\pi_1^{k},\pi_2^{k}, \ldots ,\pi_{i-1}^{k},\pi_{i+1}^{k},\pi_{i}^{k},\pi_{i+2}^{k},\ldots ,\pi_{n}^{k})$$
and
$$\bm{y}^{k+1} = (y_1^{k},y_2^{k}, \ldots ,
y_{i-1}^{k},y_{i+1}^{k}+1,y_{i}^{k}-1,y_{i+2}^{k}
,\ldots ,y_{n}^{k}).$$
\item[Step~4:] If $\pi^{k+1} = \pi$, the algorithm halts. Otherwise, the algorithm goes back to Step~1.
\end{listing}
\vspace{2mm}

Now, we estimate the running time of the above algorithm.
First, we list all the indices such that $y_i > y_{i+1}$ holds in $\order{n}$ time.
Step~2 is applied to any index in the list.
After Step~3, we update the list. 
This update can be done in  $\order{1}$ time, 
since possible new entries to the list are 
only $i+1$ and $i-1$.
Hence, each iteration takes $\order{1}$ time
and the algorithm has $\order{\inv(\bm{y})}$-iterations.
Recall that $\inv(\bm{y})$ is the number of intersections in a cyclic ladder lottery in $\llsetone(\pi,\bm{y})$.
Hence, we can construct a cyclic ladder lottery in $\llset(\pi,\bm{y})$ in $\order{n+\inv(\bm{y})}$ time.
Next, to obtain the root, we repeat to find parents at most $(\inv(\bm{y}))^2$ times.
Each process to find a parent takes $\order{\inv(\bm{y})}$ time, since one can find the smallest minimal left tangled triple in $\order{\inv(\bm{y})}$ time.
Hence, the claim is proved.
\qed
\end{proof}


Let $L$ be a cyclic ladder lottery in $\llsetone(\pi,\bm{x})$ and let $t=\mset{i,j,k}$ be a minimal right tangled triple in $L$.
We denote by $L(t)$ the cyclic ladder lottery obtained from $L$ by applying braid relation to $t$.
We can observe that $L(t)$ is a child of $L$ if and only if $t$ is the smallest minimal left tangled triple in $L(t)$.
This observation gives the child-enumeration algorithm shown in 
Algorithm~\ref{alg:enum-cll-children}.

First, we construct the root $L_0$ of $\llsetone(\pi,\bm{x})$ and 
call Algorithm~\ref{alg:enum-cll-children} with the argument $L_0$.
The algorithm outputs the current cyclic ladder lottery $L$, which is the argument of the current recursive call.
Next, for every minimal right tangled triple $t$ in $L$,
if $L(t)$ is a child of $L$, the algorithm calls itself with the argument $L(t)$.
Algorithm~\ref{alg:enum-cll-children} traverses the family tree of $\llsetone(\pi,\bm{y})$
and hence enumerates all the cyclic ladder lotteries in $\llsetone(\pi,\bm{y})$.
Each recursive call lists all the minimal right tangled triples.
To do that, we take $\order{\inv(\bm{y})}$ time.
For each minimal right tangled triple $t$, we check whether or not $L(t)$ is a child of $L$, as follows.
First, we construct $L(t)$ from $L$.
Next, in $L(t)$, we list all the minimal right tangled triples.
Finally, we check $t$ is the smallest one in the listed triples.
The answer is true implies that $L(t)$ is a child of $L$.
This takes $\order{\inv(\bm{y})}$ time.
Therefore, a recursive call of Algorithm~\ref{alg:enum-cll-children} takes $\order{(\inv(\bm{y}))^2}$ time.

\begin{algorithm}[tb]
  \DontPrintSemicolon
  Output $L$\\
  \ForEach{minimal right tangled triple $t$ in $L$}{
    \If{$t$ is the smallest triple in $L(t)$}{
        \textsc{Enum-CLL-Children}($L(t)$)
    }
  }
\caption{\textsc{Enum-CLL-Children}($L$)}
\label{alg:enum-cll-children}
\end{algorithm}

\begin{theorem}\label{thm:enum}
Let $\pi \in \frakS_n$ and let $\bm{y}$ be an almost optimal displacement vector of $\pi$.
After constructing the root in $\order{n+(\inv(\bm{y}))^3}$ time,
one can enumerate all the cyclic ladder lotteries in $\llsetone(\pi, \bm{y})$ in 
$\order{(\inv(\bm{y}))^2}$ delay.
\end{theorem}
%
%
\begin{proof}
After constructing the root, Algorithm~\ref{alg:enum-cll-children} traverses the whole of the family tree of $\llsetone(\pi,\bm{y})$ in $\order{(\inv(\bm{y}))^2 \msize{\llsetone(\pi,\bm{y})}}$ time.
However, the worst-case delay of Algorithm~\ref{alg:enum-cll-children} is not bounded in $\order{(\inv(\bm{y}))^2}$ time,
since the algorithm has to go back from deep recursive calls without any output.
By applying the ``prepostorder'' traversal method by Nakano and Uno~\cite{NakanoU05} to Algorithm~\ref{alg:enum-cll-children}, we can attain $\order{(\inv(\bm{y}))^2}$ delay.
The method outputs a cyclic ladder lottery with an odd (resp.\ even) depth in a family tree in the preorder (resp.\ postorder).
This bounds the delay between any two outputs in $\order{(\inv(\bm{y}))^2}$ time.
\qed
\end{proof}

The discussion to derive Theorem~\ref{thm:enum} can be applied to the set $\llsetopt(\pi,\bm{x})$ for an optimal displacement vector $\bm{x}$ of $\pi$. Hence, we have the following corollary.

\begin{corollary}
Let $\pi \in \frakS_n$ and let $\bm{x}$ be an optimal displacement vector of $\pi$.
After constructing the root in $\order{n + (\inv(\bm{x}))^3}$ time,
one can enumerate all the cyclic ladder lotteries in $\llsetopt(\pi, \bm{x})$ in 
$\order{(\inv(\bm{x}))^2}$ delay.
\end{corollary}

\section{Reconfiguration and enumeration of optimal cyclic ladder lotteries}

In this section, we consider the problem of enumerating all the optimal cyclic ladder lotteries in $\llsetopt(\pi)$,
where $\pi \in \frakS_n$.
That is, a displacement vector is not given as an input and only a permutation is given.
The formal description of the problem is shown below.

\medskip
\noindent
\textbf{Problem:} Enumeration of optimal cyclic ladder lotteries~(\enumCLL) \\
\textbf{Instance:} 
A permutation $\pi \in \frakS_n$.\\
\textbf{Output:}
All the optimal cyclic ladder lotteries in $\llsetopt(\pi)$ without duplication.
\medskip

From Lemma~\ref{lem:equiv_class}, we have the following outline of an enumeration algorithm to solve \enumCLL.
First, we enumerate all the optimal displacement vectors of a given permutation.
Next, for each optimal displacement vector, we enumerate all the optimal cyclic ladder lotteries using the algorithm in the previous section.

Therefore, in this section, we consider the enumeration problem of optimal displacement vectors.
We first consider the reconfiguration problem for the optimal displacement vectors to investigate the connectivity of their reconfiguration graph.
Utilizing the knowledge
of the reconfiguration graph, 
we design an enumeration algorithm that enumerates all the optimal displacement vectors of a given permutation.

\subsection{Reconfiguration}
\label{subsec:Reconfiguration}

Let $\bm{x}=\mvec{x_1,x_2,\ldots ,x_n}$ be an optimal displacement vector of a permutation $\pi$.
We denote the maximum and minimum elements in $\bm{x}$ by $\max(\bm{x})$ and $\min(\bm{x})$.
Let $i$ and $j$ be two indices such that 
$x_i = \max(\bm{x})$, $x_j = \min(\bm{x})$, and $x_i - x_j = n$.
If $\bm{x}$ includes two or more maximum (and minimum) values, the index $i$ (and $j$) is chosen arbitrarily.
Then, a \emph{max-min contraction}\footnote{The contraction is originally proposed by Jerrum~\cite{Jerrum85}.} $T_{ij}$ of $\bm{x}$ is a function $T_{ij}\colon \mathbb{Z}^n \to \mathbb{Z}^n$ such that $T_{ij}(\bm{x}) = (z_1,z_2,\ldots, z_n)$, where
\begin{align*}
    z_k &=
\begin{cases}
x_k-n & \text{if $k=i$,}\\
x_k+n & \text{if $k=j$,}\\
x_k & \text{otherwise.}
\end{cases}
\end{align*}

\noindent
We consider the following reconfiguration problem under the max-min contractions.

\medskip
\noindent
\textbf{Problem:} Reconfiguration of optimal displacement vectors~(\reconfDV) \\ 
\textbf{Instance:} 
A permutation $\pi \in \frakS_n$ and two optimal displacement vectors $\bm{x}$ and $\bm{x'}$ of $\pi$.\\
\textbf{Question:}
Does there exist a reconfiguration sequence 
between $\bm{x}$ and $\bm{x'}$ under max-min contractions?
\medskip

Jerrum~\cite{Jerrum85} showed 
the following theorem.
\begin{theorem}[\cite{Jerrum85}]\label{thm:reconfDV}
Any instance of \reconfDV\ is a yes-instance.
\end{theorem}

In the remaining part of this subsection, we consider the shortest version of \reconfDV.
Let $\bm{x} = \mvec{x_1,x_2,\ldots ,x_n}$ and $\bm{x'} = \mvec{x'_1,x'_2,\ldots ,x'_n}$ be two optimal displacement vectors of $\pi$.
We denote the length of a shortest reconfiguration sequence between $\bm{x}$ and $\bm{x'}$ under max-min contractions by $\optdv(\bm{x},\bm{x'})$.
Note that $\optdv(\bm{x},\bm{x'})$ is about
max-min contractions, while $\optclldv(L,L')$ is the length of a shortest reconfiguration sequence of two optimal cyclic ladder lotteries under braid relations.

For two optimal displacement vectors
$\bm{x}$ and $\bm{x'}$,
we define $\symdiffDV{\bm{x}}{\bm{x'}}$ by  
\begin{align*}
\symdiffDV{\bm{x}}{\bm{x'}} &= \sum_{i \in [n]} 1 - \delta_{x_i,x'_i},
\end{align*}
where $\delta_{x_i,x'_i}$ is the Kronecker delta.

One can characterize 
the length of a shortest reconfiguration sequence
using the symmetric difference, as stated in the following theorem.
\begin{theorem}\label{thm:shortest}
Let $\bm{x}$ and $\bm{x'}$ be two optimal displacement vectors of a permutation $\pi \in \frakS_n$.
Then $\optdv(\bm{x},\bm{x'}) = \frac{\symdiffDV{\bm{x}}{\bm{x'}}}{2}$ holds.
Moreover, one can compute a reconfiguration sequence of length $\optdv(\bm{x},\bm{x'})$ in $\order{n+\optdv(\bm{x},\bm{x'})}$ time.
\end{theorem}
%
%
\begin{proof}
For any max-min contraction, we have 
$\symdiffDV{\mmcont_{ij}(\bm{x})}{\bm{x'}} \geq \symdiffDV{\bm{x}}{\bm{x'}}-2$. Hence, $\optdv(\bm{x},\bm{x'}) \geq \frac{\symdiffDV{\bm{x}}{\bm{x'}}}{2}$ holds.
It is easy to find two indices $i,j$ such that $\symdiffDV{\mmcont_{ij}(\bm{x})}{\bm{x'}} = \symdiffDV{\bm{x}}{\bm{x'}}-2$.
Hence, we can construct a reconfiguration sequence of length $\frac{\symdiffDV{\bm{x}}{\bm{x'}}}{2}$.
Now, we describe how to compute a shortest reconfiguration sequence between $\bm{x}=\mvec{x_1,x_2,\ldots ,x_n}$ and $\bm{x'}=\mvec{x'_1,x'_2,\ldots ,x'_n}$ and estimate the running time of it, below.
Let $I$ be the set of indices such that, for $a \in I$, $x_a \neq x'_a$ holds.
Note that, for $a \in I$, either (1) $x_a = \max(\bm{x})$ and $x'_a = \min(\bm{x'})$ or (2) $x_a = \min(\bm{x})$ and $x'_a = \max(\bm{x'})$ holds.
Next, we partition $I$ into the set $P$ of unordered pairs $\mset{a,b}$
such that (1) $x_a=\max(\bm{x})$ and $x'_a = \min(\bm{x'})$ and (2) $x_b=\min(\bm{x})$ and $x'_b=\max(\bm{x'})$ hold.
This takes $\order{n}$ time.
Finally, we repeat to apply $\mmcont_{ab}$ for each $\mset{a,b} \in P$ until $\bm{x'}$ is obtained.
Note that applying $\mmcont_{ab}$ to $\bm{x}$ decreases $\optdv(\bm{x},\bm{x'})$ by 1 and each max-min contraction takes $\order{1}$ time.
Therefore, the claim holds.
\qed
\end{proof}

\subsection{Enumeration}

In this subsection, we consider the following enumeration problem.

\medskip
\noindent
\textbf{Problem:} Enumeration of optimal displacement vectors~(\enumDV) \\
\textbf{Instance:} 
A permutation $\pi \in \frakS_n$.\\
\textbf{Output:}
All the optimal displacement vectors of $\pi$ without duplication.
\medskip

\noindent
Theorem~\ref{thm:reconfDV} implies that the reconfiguration graph of the optimal displacement vectors of a permutation under max-min contractions is connected.
Therefore, we may use the reverse search technique to enumerate them.

Let $\pi \in \frakS_n$.
Recall that $\setodv(\pi)$ denotes the set of all the optimal displacement vectors of $\pi$.
Let $\bm{x} = \mvec{x_1,x_2,\ldots ,x_n}$ and $\bm{x'} = \mvec{x'_1,x'_2,\ldots ,x'_n}$ be two distinct optimal displacement vectors in $\setodv(\pi)$.
The vector $\bm{x}$ is \emph{larger than} $\bm{x'}$ if $x_i = x'_i$ for $i=1,2,\ldots ,j-1$
and $x_{j} > x'_{j}$.
Intuitively, $\bm{x}$ is lexicographically larger than $\bm{x'}$.
Then, note that $x_j = \max(\bm{x})$ and $x'_j=\min(\bm{x'})$ hold,
since $x_i = x'_i$ for $i=1,2,\ldots ,n$ holds if $x_i \neq \max(\bm{x})$ and $x_i \neq \min(\bm{x})$.

The \emph{root} of $\setodv(\pi)$, denoted by $\rootdv$, is the largest displacement vector among $\setodv(\pi)$.
Intuitively, the root includes the maximum values in early indices. 
Note that $\rootdv$ is unique in $\setodv(\pi)$.
Let $\bm{x}$ be an optimal displacement vector in $\setodv(\pi) \setminus \mset{\rootdv}$.
Let $\minidx(\bm{x})$ be the minimum index of $\bm{x}$
such that $\bm{x}_{\minidx(\bm{x})} = \min(\bm{x})$.
Let $\maxidx(\bm{x})$ be the minimum index of $\bm{x}$
such that $\minidx(\bm{x}) < \maxidx(\bm{x})$ and $\bm{x}_{\maxidx(\bm{x})} = \max(\bm{x})$ hold.
Then, we define the \emph{parent} of $\bm{x}$
by 
$\parentdv(\bm{x}) = \mmcont_{\maxidx(\bm{x})\minidx(\bm{x})}(\bm{x})$.
Note that $\parentdv(\bm{x})$ is larger than $\bm{x}$ and always exists for $\bm{x} \neq \rootdv$.
We say that $\bm{x}$ is a \emph{child} of $\parentdv(\bm{x})$.
The \emph{parent sequence} $\mseq{\bm{x_1},\bm{x_2}, \ldots ,\bm{x_k}}$ of $\bm{x}$ is a sequence of optimal displacement vectors in $\setodv(\pi)$ such that 
\begin{listing}{aaa}
\item[(1)] $\bm{x_1} = \bm{x}$,
\item[(2)] $\bm{x_i}=\parentdv(\bm{x_{i-1}})$ for each $i=2,3,\ldots ,m$, and
\item[(3)] $\bm{x_k} = \rootdv$.
\end{listing}
Note that one can observe that, by repeatedly finding the parents from any optimal displacement vector in $\setodv(\pi)$, the root $\rootdv$ is always obtained.
Hence,
by merging the parent sequence of 
every vector in $\bm{x} \in \setodv(\pi) \setminus \mset{\rootdv}$,
we have the tree structure rooted at $\rootdv$.
We call the tree the \emph{family tree} of $\setodv(\pi)$.
Note that the family tree is a spanning tree of the reconfiguration graph of $\setodv(\pi)$ under max-min contractions.
Therefore, to enumerate all the optimal displacement vectors in $\setodv(\pi)$, we traverse the family tree of $\setodv(\pi)$.
To traverse the family tree, we design
an algorithm to enumerate all the children of an optimal displacement vector in $\setodv(\pi)$.

\begin{algorithm}[t]
  \DontPrintSemicolon
  Output $\bm{x}$\\
  Let $m(\bm{x}) = \mseq{m_1,m_2,\ldots ,m_{\ell}}$ be the max-min index sequence of $\bm{x}$\\
  Let $m_p = \minidx(\bm{x})$ and $m_q = \maxidx(\bm{x})$\\
  \ForEach{$j = p,p+1,\ldots ,q-1$}{
    \textsc{Enum-DV-Children}($\mmcont_{m_{p-1}m_j}(\bm{x})$)
  }
\caption{\textsc{Enum-DV-Children}($\bm{x}$)}
\label{alg:enum-odv-children}
\end{algorithm}
Let $\bm{x}=\mvec{x_1,x_2,\ldots ,x_n}$ be an optimal displacement vector in $\setodv(\pi)$.
The \emph{max-min index sequence}, denoted by $m(\bm{x}) = \mseq{m_1,m_2,\ldots ,m_{\ell}}$, of $\bm{x}$ is a sequence of indices of $\bm{x}$ such that either $x_{m_i} = \max(\bm{x})$ or $x_{m_i} = \min(\bm{x})$ for $i=1,2,\ldots ,\ell$ and $m_{i} < m_{i+1}$ for each $i=1,2,\ldots ,\ell-1$.
It can be observed that if $x_{m_1} = \min(\bm{x})$, $\bm{x}$ has no child from the definition of the parent.
Hence, we assume that $x_{m_1} = \max(\bm{x})$, below.
Now, we enumerate all the children of $\bm{x}$ as follows.
Suppose that $m_p = \minidx(\bm{x})$ and $m_q = \maxidx(\bm{x})$.
(For the root $\bm{x_0}$, $\maxidx(\bm{x})$ is not defined.
Hence, for convenience, we define $\maxidx(\bm{x}) = \ell +1$ for the root.)
\begin{lemma}\label{lem:children}
Let $\bm{x}$ be an optimal displacement vector of $\pi \in \frakS_n$.
Let $m(\bm{x}) = \mseq{m_1,m_2,\ldots ,m_{\ell}}$ be the max-min index sequence of $\bm{x}$.
Then, $\mmcont_{m_i m_j}(\bm{x})$ is a child of $\bm{x}$ if and only if $i=p-1$ and $j=p,p+1,\ldots ,q-1$ hold.
\end{lemma}
%
%
\begin{proof}
We have the following case analysis for two indices~$i$ and $j$~$(1<i<j<\ell)$.

\begin{mycase}{1}{$i<p-1$.}
For any $j$ with $i<j$,
$\mmcont_{m_i m_j}(\bm{x})$ is not a child,
since $\bm{x} \neq \parent(T_{m_i m_j}(\bm{x}))$ holds.
\end{mycase}

\begin{mycase}{2}{$i=p-1$.}
For each $j=p,p+1,\ldots ,q-1$,
$\mmcont_{m_i m_j}(\bm{x})$ is a child, 
since $\bm{x} = \parent(T_{m_i m_j}(\bm{x}))$ holds.
Otherwise, if $q-1<j$, $T_{m_i m_j}(\bm{x})$ is not a child.
\end{mycase}

\begin{mycaselast}{3}{$p-1<i$.}
For any $i$ and $j$ with $i<j$, $T_{m_i m_j}(\bm{x})$ is not a child,
since $\bm{x} \neq \parent(\mmcont_{m_i m_j}(\bm{x}))$ holds.
\end{mycaselast}

From the above case analysis,
if $\mmcont_{m_im_j}(\bm{x})$ is a child of $\bm{x}$,
$i=p-1$ and $j=p,p+1,\ldots ,q-1$ hold.
Moreover, from Case~2, if $i=p-1$ and $j=p,p+1,\ldots ,q-1$ hold,
$\mmcont_{m_im_j}(\bm{x})$ is a child of $\bm{x}$.
\qed
\end{proof}

From Lemma~\ref{lem:children}, we have the child-enumeration algorithm shown in 
Algorithm~\ref{alg:enum-odv-children}.
We first construct the root $\bm{x_0}$ and call the algorithm with the argument $\bm{x_0}$.
By recursively calling Algorithm~\ref{alg:enum-odv-children}, one can traverse the family tree.

\begin{theorem}\label{thm:enum_dv}
Let $\pi \in \frakS_n$.
After $\order{n^2}$-time preprocessing, one can enumerate all the optimal displacement vectors in $\setodv(\pi)$ in 
$\order{1}$ delay.
\end{theorem}
%
%
\begin{proof}
We construct $\rootdv$ of $\setodv(\pi)$, as follows.
First, we construct an optimal displacement vector $\bm{x}$ using the algorithm by Jerrum~\cite{Jerrum85}, which takes $\order{n^2}$ time. 
If $\bm{x} \neq \rootdv$, 
we construct $\rootdv$ by repeatedly finding the parents. 
This takes $\order{n}$ time.

In Algorithm~\ref{alg:enum-odv-children}, we enumerate all the children of an optimal displacement vector $\bm{x}$.
This can be done in $\order{1}$ time for each child, as follows.
If we have the max-min index sequence of $\bm{x}$, $\minidx(\bm{x})$ and $\maxidx(\bm{x})$,
Algorithm~\ref{alg:enum-odv-children} can generate each child in $\order{1}$ time.
We can maintain this information, as follows.
Suppose that $\mmcont_{m_i m_j}(\bm{x})$ is a child of $\bm{x}$.
Then, the max-min sequence of $\mmcont_{m_i m_j}(\bm{x})$ is the same as the one of $\bm{x}$.
Moreover, $\minidx(\mmcont_{m_i m_j}(\bm{x})) = m_i$ and $\maxidx(\mmcont_{m_i m_j}(\bm{x})) = m_j$ hold.
Hence, we can update the information in $\order{1}$ time.
(For the root $\bm{x_0}$, we construct the max-min index sequence, $\minidx(\bm{x_0})$, and $\maxidx(\bm{x_0})$ in $\order{n}$ time.)

Finally, by applying the ``prepostorder'' traversal method by Nakano and Uno~\cite{NakanoU05}, we attain $\order{1}$ delay.
\qed
\end{proof}

\section*{Acknowledgments}

This work is partially supported by JSPS KAKENHI
Grant Numbers
JP22H03549, 
JP22K17849, 
JP23K11027, 
JP20K14317, JP23K12974, and JP24H00686. 

\bibliographystyle{elsarticle-num} 
\bibliography{mybib}

\providecommand*\hyphen{-}
\begin{thebibliography}{10}
\expandafter\ifx\csname url\endcsname\relax
  \def\url#1{\texttt{#1}}\fi
\expandafter\ifx\csname urlprefix\endcsname\relax\def\urlprefix{URL }\fi
\expandafter\ifx\csname href\endcsname\relax
  \def\href#1#2{#2} \def\path#1{#1}\fi

\bibitem{NozakiWY23}
Y.~Nozaki, K.~Wasa, K.~Yamanaka, Reconfiguration and enumeration of optimal
  cyclic ladder lotteries, in: Proc. of The 34th International Workshop on
  Combinatorial Algorithms~({IWOCA}2023), Vol. 13889 of Lecture Notes in
  Computer Science, Springer, 2023, pp. 331--342.

\bibitem{Knuth92}
D.~E. Knuth, Axioms and hulls, LNCS 606, Springer-Verlag, 1992.

\bibitem{YamanakaNMUN10}
K.~Yamanaka, S.~Nakano, Y.~Matsui, R.~Uehara, K.~Nakada, Efficient enumeration
  of all ladder lotteries and its application, Theoretical Computer Science 411
  (2010) 1714--1722.

\bibitem{F97}
S.~Felsner, On the number of arrangements of pseudolines, Discrete \&
  Computational Geometry 18 (1997) 257--267.

\bibitem{FelsnerV11}
S.~Felsner, P.~Valtr, Coding and counting arrangements of pseudolines, Discrete
  \& Computational Geometry 46 (2011) 405--416.

\bibitem{DumitrescuM20}
A.~Dumitrescu, R.~Mandal, New lower bounds for the number of pseudoline
  arrangements, J. Comput. Geom. 11~(1) (2020) 60--92.

\bibitem{KawaharaSYM11}
J.~Kawahara, T.~Saitoh, R.~Yoshinaka, S.~ichi Minato, Counting primitive
  sorting networks by $\pi${DD}s, Hokkaido University, Division of Computer
  Science, TCS Technical Reports TCS-TR-A-11-54 (2011).

\bibitem{Samuel11}
M.~J. Samuel, Word posets, with applications to {C}oxeter groups, in:
  Proceedings of The 8th International Conference WORDS 2011, Vol.~63 of EPTCS,
  2011, pp. 226--230.

\bibitem{Sloane22}
N.~J.~A. Sloane, The on-line encyclopedia of integer sequences, Published
  electronically at \url{https://oeis.org/A006245}, accessed: 2022-09-25.

\bibitem{MKS04}
W.~Magnus, A.~Karrass, D.~Solitar, Combinatorial group theory, 2nd Edition,
  Dover Publications, Inc., Mineola, NY, 2004, presentations of groups in terms
  of generators and relations.

\bibitem{Lus83}
G.~Lusztig, Some examples of square integrable representations of semisimple
  {$p$}-adic groups, Trans. Amer. Math. Soc. 277~(2) (1983) 623--653.

\bibitem{Jerrum85}
M.~R. Jerrum, The complexity of finding minimum-length generator sequence,
  Theoretical Computer Science 36 (1985) 265--289.

\bibitem{YamanakaHW21}
K.~Yamanaka, T.~Horiyama, K.~Wasa, Optimal reconfiguration of optimal ladder
  lotteries, Theoretical Computer Science 859 (2021) 57--69.

\bibitem{BjBr05}
A.~Bj\"{o}rner, F.~Brenti, Combinatorics of {C}oxeter groups, Vol. 231 of
  Graduate Texts in Mathematics, Springer, New York, 2005.

\bibitem{ZuylenBSY16}
A.~van Zuylen, J.~Bieron, F.~Schalekamp, G.~Yu,
  \href{https://doi.org/10.1016/j.ipl.2016.06.006}{A tight upper bound on the
  number of cyclically adjacent transpositions to sort a permutation}, Inform.
  Process. Lett. 116~(11) (2016) 718--722.
\newline\urlprefix\url{https://doi.org/10.1016/j.ipl.2016.06.006}

\bibitem{AvisF96}
D.~Avis, K.~Fukuda, Reverse search for enumeration, Discrete Applied
  Mathematics 65~(1-3) (1996) 21--46.

\bibitem{NakanoU05}
S.~Nakano, T.~Uno, Generating colored trees, Proceedings of the 31th Workshop
  on Graph-Theoretic Concepts in Computer Science, (WG 2005) LNCS 3787 (2005)
  249--260.

\end{thebibliography}

\end{document}